\documentclass[runningheads]{llncs}

\usepackage[T1]{fontenc}
\usepackage{cite}
\usepackage{amsfonts}
\usepackage[ruled,vlined]{algorithm2e}
\usepackage{algpseudocode}

\usepackage{graphicx}
\usepackage{textcomp}
\usepackage{xcolor}
\usepackage{microtype}
\usepackage{graphicx}
\usepackage{subfigure}
\usepackage{booktabs} 
\usepackage{paralist}
\usepackage[utf8]{inputenc}

\usepackage[shortlabels]{enumitem}
\usepackage{lipsum, adjustbox}
\usepackage{multirow}
\usepackage{multicol}

\usepackage{pgf-umlsd}
\usepackage{breqn}
\usepackage{etoolbox}
\AtBeginEnvironment{figure}{\setlength{\textfloatsep}{8pt}\setlength{\intextsep}{8pt}\setlength{\abovecaptionskip}{4pt}\setlength{\belowcaptionskip}{4pt}}

\usepackage{color,soul}
\usepackage[font=scriptsize,labelfont=bf]{caption}
\usepackage{balance}
\usepackage{lipsum, adjustbox}
\usepackage[colorlinks,breaklinks]{hyperref} 
\hypersetup{
	hypertexnames=true, linkcolor=blue, anchorcolor=black,
	citecolor=magenta, urlcolor=alizarin
}

\newsavebox\myboxA
\newsavebox\myboxB
\newlength\mylenA
\newcommand{\overbar}[1]{\mkern 1.5mu\overline{\mkern-1.5mu#1\mkern-1.5mu}\mkern 1.5mu}

\AtBeginDocument{%
  \providecommand\BibTeX{{%
    \normalfont B\kern-0.5em{\scshape i\kern-0.25em b}\kern-0.8em\TeX}}}

\usepackage{color, colortbl}
\definecolor{Gray}{gray}{0.9}
\definecolor{airforceblue}{rgb}{0.36, 0.54, 0.66}
\definecolor{aliceblue}{rgb}{0.94, 0.97, 1.0}
\definecolor{alizarin}{rgb}{0.82, 0.1, 0.26}
\definecolor{amber}{rgb}{1.0, 0.75, 0.0}
\definecolor{amber(sae/ece)}{rgb}{1.0, 0.49, 0.0}
\definecolor{bronze}{rgb}{0.8, 0.5, 0.2}
\definecolor{battleshipgrey}{rgb}{0.52, 0.52, 0.51}
\definecolor{bole}{rgb}{0.47, 0.27, 0.23}
\definecolor{bulgarianrose}{rgb}{0.28, 0.02, 0.03}
\definecolor{cadet}{rgb}{0.33, 0.41, 0.47}
\definecolor{ceil}{rgb}{0.57, 0.63, 0.81}
\definecolor{cerulean}{rgb}{0.0, 0.48, 0.65}
\definecolor{charcoal}{rgb}{0.21, 0.27, 0.31}
\definecolor{coolblack}{rgb}{0.0, 0.18, 0.39}
\definecolor{coolgrey}{rgb}{0.55, 0.57, 0.67}
\definecolor{darkcandyapplered}{rgb}{0.64, 0.0, 0.0}
\definecolor{darkbrown}{rgb}{0.4, 0.26, 0.13}
\definecolor{darkcerulean}{rgb}{0.03, 0.27, 0.49}
\definecolor{darkgray}{rgb}{0.66, 0.66, 0.66}
\definecolor{darkjunglegreen}{rgb}{0.1, 0.14, 0.13}
\definecolor{darktaupe}{rgb}{0.28, 0.24, 0.2}
\definecolor{davy\'sgrey}{rgb}{0.33, 0.33, 0.33}
\definecolor{frenchblue}{rgb}{0.0, 0.45, 0.73}
\definecolor{almond}{rgb}{0.94, 0.87, 0.8}
\definecolor{beaublue}{rgb}{0.74, 0.83, 0.9}
\definecolor{beige}{rgb}{0.96, 0.96, 0.86}
\definecolor{bisque}{rgb}{1.0, 0.89, 0.77}
\definecolor{black}{rgb}{0.0, 0.0, 0.0}
\definecolor{fluorescentorange}{rgb}{1.0, 0.75, 0.0}
\definecolor{ghostwhite}{rgb}{0.97, 0.97, 1.0}
\definecolor{antiquewhite}{rgb}{0.98, 0.92, 0.84}
\definecolor{LightCyan}{rgb}{0.88,1,1}

\usepackage[textsize=tiny]{todonotes}

\def\BibTeX{{\rm B\kern-.05em{\sc i\kern-.025em b}\kern-.08em
    T\kern-.1667em\lower.7ex\hbox{E}\kern-.125emX}}
\let\llncssubparagraph\subparagraph
\let\subparagraph\paragraph
\usepackage[compact]{titlesec}
\let\subparagraph\llncssubparagraph
\usepackage[compact]{titlesec}         
\titlespacing{\section}{0pt}{0pt}{0pt} 
\AtBeginDocument{
	\setlength\abovedisplayskip{0pt}
	\setlength\belowdisplayskip{0pt}}
\begin{document}
%
\title{A More Secure Split: Enhancing the Security of Privacy-Preserving Split Learning\thanks{This work was funded by the Technology Innovation Institute (TII) for the project ARROWSMITH and from Horizon Europe for HARPOCRATES (101069535).}
}
%
\titlerunning{A More Secure Split}
\author{Tanveer Khan\inst{1}\orcidID{0000-0001-7296-2178} \and Khoa Nguyen\inst{1}\orcidID{0000-0001-5089-5250}  \and
Antonis Michalas\inst{1,2}\orcidID{0000-0002-0189-3520}}

\authorrunning{T. Khan et al.}

\institute {Tampere University, Finland \and RISE Research Institutes of Sweden \\
\email{\{tanveer.khan, khoa.nguyen, antonios.michalas\}@tuni.fi}}
%
\maketitle              
\begin{abstract}
Split learning (SL) is a new collaborative learning technique that allows participants, e.g.\ a client and a server, to train machine learning models without the client sharing raw data. In this setting, the client initially applies its part of the machine learning model on the raw data to generate Activation Maps (AMs) and then sends them to the server to continue the training process. Previous works in the field demonstrated that reconstructing AMs could result in privacy leakage of client data. In addition to that, existing mitigation techniques that overcome the privacy leakage of SL prove to be significantly worse in terms of accuracy. In this paper, we improve upon previous works by constructing a protocol based on U-shaped SL that can operate on homomorphically encrypted data. More precisely, in our approach, the client applies homomorphic encryption on the AMs before sending them to the server, thus protecting user privacy. This is an important improvement that reduces privacy leakage in comparison to other SL-based works. Finally, our results show that, with the optimum set of parameters, training with HE data in the U-shaped SL setting only reduces accuracy by 2.65\% compared to training on plaintext. In addition, raw training data privacy is preserved.

\keywords{Activation Maps \and Homomorphic Encryption \and Machine Learning \and Privacy \and Split Learning}
\end{abstract}
\section{Introduction}


Nowadays, machine learning (ML) methods are widely used in many applications due to their predictive and generative power. However, this raises serious concerns regarding user data privacy, leading to the need for privacy-preserving machine learning (PPML) solutions~\cite{khan2021blind}. Split Learning (SL) and Federated Learning (FL) are two PPML methods that rely on training ML models on decentralized data sources~\cite{vepakomma2019reducing}. In FL~\cite{yang2019federated}, every client runs a copy of the entire model on its data. The server receives updated weights from each client and aggregates them. The SL~\cite{gupta2018distributed} model divides the Neural Network (NN) into two parts: the client-side and the server-side. SL is used for training NN among multiple data sources, while mitigating the need to directly share raw labeled data with collaboration parties. 
The advantages of SL are multifold: \begin{inparaenum}[\it (i)] \item it allows multiple parties to collaboratively train a NN, \item it allows users to train ML models without sharing their raw data with a server running part of a NN model, thus preserving user privacy, \item it protects both the client and the server from revealing their parts of the model, and \item it reduces the client's computational overhead by not running the entire model~\cite{vepakomma2019split}.\end{inparaenum}

Though SL offers an extra layer of privacy protection by definition, there are no works exploring how it is combined with popular techniques that promise to preserve user privacy (e.g.\ encryption). In~\cite{abuadbba2020can}, the authors studied whether SL can handle sensitive time-series data and demonstrated that SL alone is \textit{insufficient} when performing privacy-preserving training for 1-dimensional (1D) CNN models. More precisely, the authors showed raw data can be reconstructed from the AMs of the intermediate split layer. The authors also employed two mitigation techniques, adding hidden layers and applying differential privacy to reduce privacy leakage. However, based on the results, none of these techniques can effectively reduce privacy leakage from all channels of the SL activation. Furthermore, both these techniques result in significantly reducing the joint model's accuracy.

In this paper, we focus on training an ML model in a privacy-preserving manner, where a client and a server collaborate to train the model. More specifically, we construct a model that uses Homomorphic Encryption (HE)~\cite{cheon2017homomorphic} to mitigate privacy leakage in SL. In our model, the client first encrypts the AMs and then sends the Encrypted Activation Maps (EAMs) to the server. The EAMs do \textit{not} reveal anything about the raw data (i.e.\ it is \textit{not} possible to reconstruct the original raw data from the EAM).


\medskip 
\noindent \textbf{\textit{Contributions:}} \enskip
The main contributions of this paper are:  
\begin{enumerate}[\bfseries C1.]
    \item We designed a simplified version of the 1D CNN model presented in~\cite{abuadbba2020can} and we are using it to classify the ECG signals~\cite{moody2001impact} in both local and SL settings. More specifically, we construct a U-shaped split 1D CNN model and experiment using plaintext AMs sent from the client to the server. Through the U-shaped 1D CNN model, clients do \textit{not} need to share the input training samples and the ground truth labels with the server -- this is an important improvement that reduces privacy leakage compared to~\cite{abuadbba2020can}.
    \item We constructed the HE version of the U-shaped SL technique. In the encrypted U-shaped SL model, the client encrypts the AM using HE and sends it to the server. The advantage of HE encrypted U-shaped SL over the plaintext U-shaped SL is that server performs computation over EAMs.
    \item To assess the applicability of our framework, we performed experiments on two heartbeat datasets: the MIT-DB~\cite{moody2001impact} and the PTB-XL~\cite{wagner2020ptb}, with PTB-XL currently being the largest open-source electrocardiography dataset to our knowledge. For the MIT-DB dataset, we experimented with AMs of two lengths (256 and 512) for both plaintext and homomorphically EAMs and we have measured the model's performance by considering training duration test accuracy, and communication cost. We performed similar experiments with PTB-XL dataset, however, only with AMs of length 256. 
    \item Moreover, our framework takes advantage of batch encryption, an optimization technique for memory and computation, to improve computing performance over encrypted data. We conducted experiments with and without batch encryption and compared results.
    \item We designed a detailed protocol to prove our construction's and provide proof of its security level under the malicious threat model.
\end{enumerate}

\medskip 

\noindent \textbf{\textit{Organization:}} \enskip
	The rest of the paper is organized as follows\footnote{Due to space constraints, the necessary background information about 1D CNN, HE and SL are in~\autoref{sec:preliminaries}.}: 
	In \autoref{sec:relatedwork}, we present important published works in the area of SL. 
	The architectures of the proposed models are presented in \autoref{sec:architecture}. 
	The design and implementation of split 1D CNN training protocols are described in \autoref{sec:slProtocol}, formal protocol construction in \autoref{sec:protocol}, protocol security in \autoref{app:secanalysis}, extensive experimental results in \autoref{sec:performance},
and conclude the paper in~\autoref{sec:Conclusion}.

\section{Related Work}
\label{sec:relatedwork}
One of the primary reasons researchers seek novel techniques is to bridge the large gap between existing privacy solutions and the actual practical deployment of NNs. 
PPML consists of cryptographic approaches such as HE and Multiparty Computation (MPC), differential privacy as well as distributed ML approaches such as FL, and SL~\cite{cabrero2021sok}. \href{https://ai.googleblog.com/2017/04/federated-learning-collaborative.html}{Google AI Blog}
introduced FL, where users (e.g. mobile devices) collaboratively train a model~\cite{mcmahan2017communication} under a central server's orchestration (e.g. service provider) without sharing their data. However, in FL, sharing user model weights with server can lead to sensitive information leaks~\cite{hitaj2017deep}.

	

SL approach~\cite{gupta2018distributed} 
is a promising approach in terms of client raw data protection, however, SL provides data privacy on the grounds that only intermediate AMs are shared between the parties. Different studies showed the possibility of privacy leakage in SL. In~\cite{vepakomma2019reducing}, the authors analyzed the privacy leakage of SL and found a considerable leakage from the split layer in the 2D~CNN model. Furthermore, the authors mentioned that it is possible to reduce the distance correlation (a measure of dependence) between the split layer and raw data by slightly scaling the weights of all layers before the split. This scaling works well in models with a large number of hidden layers before the split. 

The work of Abuadbba \textit{et al.}~\cite{abuadbba2020can} is the first study exploring whether SL can deal with time-series data. The authors proved that only SL cannot preserve the privacy of the data, and employed two techniques to resolve this privacy problem. However, both suffer from a loss of model accuracy, with the use of \textit{differential privacy degrading the classification accuracy significantly from 98.9\% to 50\%.}

\section{Architecture}
\label{sec:architecture}
In this section, we first describe the non-split version or local model of the 1D CNN used to classify the ECG signal. Then, we discuss the process of splitting this local model into a U-shaped split model. Furthermore, we also describe the involved parties (a client and a server) in the training process of the split model, focusing on their roles and the parameters assigned to them throughout the training process. Notations for all parameters and their descriptions is in \autoref{table: paranddes}.

\begin{table}
    \centering
    \caption{Parameters and Description in the Algorithms}
    \resizebox{\textwidth}{!}{%
    \label{table: paranddes}
    \scriptsize
    \begin{tabular}{l|l|l|l|l}
        \hline
        \#	& ML & Description & HE & Description\\
        &Parameters&&Parameters&\\   \hline
        1	& \textit{D}  & Dataset & $\mathcal{P}$ & Polynomial modulus\\ \hline
        2	& $\mathbf{x}$, $\mathbf{y}$ & Input data samples and ground-truth labels &  $\mathcal{C}$ & Coefficient modulus\\ \hline
        3	& \textit{n, N} & Batch size and number of batches to be trained & $\Delta$  & Scaling factor\\ \hline
        4	& $\boldsymbol{w}^{i}, \boldsymbol{b}^{i}$ & Weights and biases in layer \textit{i}& CKKS  &Encryption scheme \\ \hline
        5	& $f^{i}$ & Linear or convolution operation of layer \textit{i}& $\mathsf{pk}$ & Public key \\ \hline
        6	& $g^{i}$ & Activation function of layer \textit{i} & $\mathsf{sk}$  & Secret key \\  \hline
        7	& $\mathbf{a}^{i}$ & Output activation maps of $g^{i}$ & $\mathsf{HE.Enc}$  & Homomorphic encryption  \\ \hline
        8	& $\mathbf{z}^{i}$ & Output tensor of $f^{i}$ & $\mathsf{HE.Dec}$  & Homomorphic decryption \\ \hline
        9	& $\eta$ & Learning rate &$\mathsf{HE.Eval}$ & Homomorphic evaluation \\  \hline
        10	& $\boldsymbol{\Phi}$ & Model's weights & $\mathsf{ctx_{pri}}$ & Private context \\ \hline
        11	&$E$& Number of training epochs  & $\mathsf{ctx_{pub}}$ & Public context \\  \hline
        12	& $\mathcal{L}$, $J$ & Loss function and error & $\bar{\mathbf{a}}^{i}$  & Encrypted activation maps    \\ \hline
        13	& $O$ & Optimizer  &  $\bar{\mathbf{z}}^{i}$ &  Encrypted tensor \\ \hline
    \end{tabular}}
\end{table}

\subsection{1D CNN Local Model Architecture}
\label{subsec:local_model}
We first implement and successfully reproduce the local model results~\cite{abuadbba2020can}. This model contains two Conv1D layers and two FC layers. The optimal test accuracy that this model achieves is 98.9\%. We implement a simplified version where the model has one less FC layer compared to the model from~\cite{abuadbba2020can}. Our local model consists of all the layer of~\autoref{fig:u-shapedSL} without any split between the client and the server. As can be seen in \autoref{fig:u-shapedSL}, we limit our model to two Conv1D layers and one linear layer as we aim to reduce computational costs when HE is applied on AMs in the model's split version. Reducing the number of FC layers leads to a drop in the accuracy of the model. The best test accuracy we obtained after training our local model for 10 epochs with a \textit{n} of 4 is 92.84\%. \textit{Although reducing the number of layers affects the model's accuracy, it is not within our goals to demonstrate how successful our ML model is for this task; instead, our focus is to construct a split model where training and evaluation on encrypted data are comparable to training and evaluation on plaintext data.} We also apply the simplified 1D CNN  on the PTB-XL dataset, with a small modification due to the difference in the number of input channels compared to the dataset from~\cite{moody2001impact}. The training result on the PTB-XL dataset after 10 epochs with a \textit{n} of 4 is 74.01\%, with the best test accuracy of 67.36\%. In \autoref{sec:performance}, we detail results for the non-split version and compare them with split version.

The training process of the local 1D CNN can be described as following: Suppose we have a heartbeat data sample $\mathbf{x} \in \mathbb{R}^{c}$, where $c$ is the number of input features or the number of time steps. $\mathbf{x}$ belongs to one out of $m$ ground-truth classes. Each data sample $\mathbf{x}$ has a corresponding encoded label vector  $\mathbf{y}\in \mathbb{R}^{m}$ that represents its ground-truth class. We can write the 1D CNN as a function $f_{\boldsymbol{\Phi}}$, where $\boldsymbol{\Phi}$ is a set of adjustable parameters as denoted in~\autoref{table: paranddes}. $\Phi$ is first initialized to small random values in the range $[-1, 1]$. Our aim is to find the best set of parameters to map $\mathbf{x}$ to a predicted output vector $\mathbf{\hat{y}} \in \mathbb{R}^{m}$, where $\mathbf{\hat{y}}$ is as close as possible to $\mathbf{y}$. $\mathbf{\hat{y}}$ can be a vector of $m$ probabilities, and we pick $\mathbf{x}$ to belong to the class with the highest probability. To find the closest value of $\mathbf{\hat{y}}$ with respect to $\mathbf{y}$, we try to minimize a loss function $\mathcal{L}({\mathbf{\hat{y}}}, \mathbf{y})$. Training the 1D CNN is an iterative process to find the best $\boldsymbol{\Phi}$ to minimize the loss function. This process consists of two sub-processes called ``forward propagation'' and ``backward propagation''. More specifically, forward propagation moves from the input $\mathbf{x}$ throughout the network, reaching the output layer and produces the predicted output $\mathbf{\hat{y}}$. Conversely, backward propagation moves from the network's output layer back to the input layer to calculate the gradients of the loss function $\mathcal{L}$ w.r.t the weights $\boldsymbol{\Phi}$ of the network. These weights are then updated according to the gradients. The process of calculating the predicted output, the loss function, the gradients and then updating the weights is called ``training''. We train the NN with thousands of samples of $\mathbf{x}$'s and corresponding $\mathbf{y}$'s, through many iterations of forward and backward propagation. We do not train the network on each single data example, but use a number of them at a time (defined by the \textit{n}). The total number of training batches is $N=\frac{|D|}{n}$, where $|D|$ is the size of the dataset. Once the NN goes through all the training batches, it has completed one training epoch. This process repeats for $E$ epochs in total. 

\subsection{U-shaped Split 1D CNN Model}
\label{subsec:u_shaped_split}

In this section, we first present the constructed U-shaped split model. We then report in more detail the roles and access rights of the actors who are involved in the training protocols of the split 1D CNN on both plaintexts and EAMs. The SL protocol consists of two parties: the client and server. We split the local 1D CNN into multiple parts, where each party trains its part(s) and communicates with 
others to complete the overall training procedure. More specifically, we construct the U-shaped split 1D CNN in such a way that the first few as well as the last layer are on the client-side, while the remaining layers are on the server-side, as demonstrated in \autoref{fig:u-shapedSL}. The client and server collaborate to train the split model by sharing the AMs and gradients. On the client-side, there are two Conv1D, two Max Pooling, two Leaky ReLU layers, and a Softmax layer. On server-side, there is only one linear layer. As mentioned earlier, the reason for having only one linear layer on the server-side is due to computational constraints when training on encrypted data.

                

\subsection{Actors in the Split Learning Model}
\label{subsec:slActors}
As mentioned earlier, in our SL setting, we have two involved parties: the client and the server. Each party plays a specific role and has access to certain parameters. More specifically, their roles and accesses are described as
\begin{itemize}
    \item Client: In the plaintext version, the client holds two Conv1D layers and can access their weights and biases in plaintext. 
    In the HE encrypted version, the client generates the HE context and has access to all context parameters ($\mathcal{P}$, $\mathcal{C}$, $\Delta$, $\mathsf{pk}$ and $\mathsf{sk}$). Note that for both training on plaintext and EAMs, the raw data examples $\mathbf{x}$'s and their corresponding labels $\mathbf{y}$'s reside on the client side and are never sent to the server during the training process.
    \item Server: In our model, the computation performed on the server-side is limited to only one linear layer. Hence, the server can exclusively access the weights and biases of this linear layer. The server also has access to the HE parameters except for the secret key $\mathsf{sk}$. 
    The hyperparameters shared between the client and the server are $\eta$, $n$, $N$, $E$.
\end{itemize}

\section{Split Model Training Protocols}
\label{sec:slProtocol}
We first present the protocol for training the U-shaped split 1D CNN on plaintext AMs, followed by training the U-shaped split 1D CNN on EAMs.

\subsection{Plaintext Activation Maps}
\label{subsection: unencryptedactivation}
We have used ~\autoref{alg:client} and~\autoref{alg:server} to train the U-shaped split 1D CNN reported in \autoref{subsec:u_shaped_split}. First, the client and server start the socket initialization process and synchronize the hyperparameters $\eta, n, N, E$. They also initialize the weights of their layers according to $\Phi$.

During forward propagation , the client forward-propagates the input $\mathbf{x}$ until $l^{th}$ layer and sends $\mathbf{a}^{(l)}$ to the server. The server continues 
to forward propagate and sends the output $\mathbf{a}^{(L)}$ to the client. Next, the client applies the Softmax function on $\mathbf{a}^{(L)}$ to get $\mathbf{\hat{y}}$ and calculates the error $J = \mathcal{L}(\mathbf{\hat{y}}, \mathbf{y})$.

The client starts the backward propagation by calculating and sending the gradient of the error w.r.t $\mathbf{a}^{(L)}$, i.e. $\frac{\partial J}{\partial \mathbf{a}^{(L)}}$, to the server. The server continues the backward propagation, calculates $\frac{\partial J}{\partial \mathbf{a}^{(l)}}$ and sends $\frac{\partial J}{\partial \mathbf{a}^{(l)}}$ to the client. After receiving the gradients $\frac{\partial J}{\partial \mathbf{a}^{(l)}}$ from the server, the backward propagation continues to the first hidden layer on the client-side. Note that the exchange of information between client and server in these algorithms takes place in plaintext. As can be seen in~\autoref{alg:client}, the client sends the AMs $\mathbf{a}^{(l)}$ to the server in plaintext and receives the output of the linear layer $\mathbf{a}^{(L)}$ from the server in plaintext. 	
The same applies on the server side: receiving $\mathbf{a}^{(l)}$ and sending $\mathbf{a}^{(L)}$ in the plaintext as can be seen in~\autoref{alg:server}. Sharif \textit{et al.}~\cite{abuadbba2020can} showed that the exchange of plaintext AMs between client and server using SL reveals important information regarding the client's raw data. Later, in~\autoref{subsec:VisualInvert} we show in detail how passing the forward AMs from client to server in the plaintext will result in information leakage. To mitigate this privacy leakage, we propose the protocol, where the client encrypts AMs before sending them to the server (see ~\autoref{subsubsection: encryptedactivation}).

\subsection{Encrypted Activation Maps}
\label{subsubsection: encryptedactivation}
The protocol for training the U-shaped 1D CNN with a homomorphically EAM consists of four phases: initialization, forward propagation, classification, and backward propagation. The initialization phase only takes place once at the beginning of the procedure, whereas the other phases continue until the model iterates through all epochs.
Each phase is explained in detail below.

\paragraph*{Initialization} This phase consists of socket initialization, context generation, and random weight loading. The client establishes a socket connection to the server and synchronizes the four hyperparameters $\eta ,\ n,\ N, E $ with the server, shown in~\autoref{alg:clientHE} and~\autoref{alg:serverHE}. These parameters must be synchronized on both sides to be trained in the same way.
Also, the weights on the client and server are initialized with the same set of corresponding weights in the local model to accurately assess and compare the influence of SL on performance. On both, client and server, $\boldsymbol{w}^{(i)}$ are initialized using corresponding parts of $\Phi$. $\mathbf{a}^{(i)}$, $\mathbf{z}^{(i)}$, and the gradients are initially set to zero. In this phase, the context generated is an object that holds $\mathsf{pk}$ and $\mathsf{sk}$ of the HE scheme as well as 
$\mathcal{P}$, $\mathcal{C}$ and $\Delta$.
\begin{minipage}{0.5\textwidth}

\begin{algorithm}[H]
\scriptsize
\SetAlgoLined
 \textbf{Initialization:}\\
 $s\leftarrow$ socket initialized with port and address\;
 \textit{s.connect}\\
 $\eta, n, N, E \leftarrow s.synchronize()$\\
 $ \{\boldsymbol{w}^{( i)}, \boldsymbol{b}^{( i)}\}_{\forall i\in \{0..l\}} \ \leftarrow \ initialize\ using\ \Phi $\\
 $\{\mathbf{z}^{( i)}\}_{\forall i\in \{0..l\}} ,\{\mathbf{a}^{( i)}\}_{\forall i\in \{0..l\}}\leftarrow \emptyset \ $\\
 $ \left\{\frac{\partial J}{\partial \mathbf{z}^{( i)}}\right\}_{\forall i\in \{0..l\}} ,\left\{\frac{\partial J}{\partial \mathbf{a}^{( i)}}\right\}_{\forall i\in \{0..l\}}\leftarrow \emptyset \ $\\
 \For{$\displaystyle e \ \in \ E $}{
 	\For{$\displaystyle \text{each} \ \text{batch}\ ( \mathbf{x},\ \mathbf{y}) \ \text{from}\ D\ $}{
 	      $\displaystyle  \mathbf{Forward\ propagation:}$
          $\displaystyle \ \ \ \ O.zero\_grad()  $\\
 	$\displaystyle \ \ \ \ \mathbf{a}^{0} \ \ \leftarrow \mathbf{x}$ \\
 	\For{$i \leftarrow 1$ to $l$}{$\displaystyle \ \ \ \ \mathbf{for} \ i\ \leftarrow \ 1\ \mathbf{to} \ l\ \mathbf{do}$\\
 	$\displaystyle \ \ \ \ \ \ \ \ \mathbf{z}^{( i)} \ \leftarrow \ f^{( i)}\left( \mathbf{a}^{( i-1)}\right)$\\
 	$\displaystyle \ \ \ \ \ \ \ \ \mathbf{a}^{( i)} \ \leftarrow \ g^{( i)}\left( \mathbf{z}^{( i)}\right)$\\}
 	$\displaystyle \ \ \ \  s.send\ ( \mathbf{a}^{(l)})$\\
 	$\displaystyle \ \ \ \ s.receive\ ( \mathbf{a}^{(L)})$\\
 	$\displaystyle \ \ \ \  \hat{y} \ \leftarrow \ Softmax\left(\mathbf{a}^{( L)}\right)$\\
 	$\displaystyle \ \ \ \ J \leftarrow \mathcal{L} (\hat{\mathbf{y}}, \mathbf{y})$\\
 	$\displaystyle \mathbf{Backward\ propagation:}$\\
 	$\displaystyle \ \ \ \ \text{Compute}\left\{\frac{\partial J}{\partial \hat{\mathbf{y}}}\ \&\ \frac{\partial J}{\partial \mathbf{a}^{(L)}}\right\}$\\
 	$\displaystyle \ \ \ \ s.send\ \left( \frac{\partial J}{\partial \mathbf{a}^{(L)}} \right)$\\
 	$\displaystyle \ \ \ \ s.receive\ \left( \frac{\partial J}{\partial \mathbf{a}^{( l)}} \right)$\\
 	\For{$i\leftarrow 1$ to $l$}{$\text{Compute}\ \left\{ \frac{\partial J}{\partial \boldsymbol{w}^{( i)}}, \ \frac{\partial J}{\partial \boldsymbol{b}^{( i)}} \right\}$\\
 	$\displaystyle\ \ \ \ \ \ \ \ \text{Update}\ \boldsymbol{w}^{( i)},\ \boldsymbol{b}^{( i)}$
 	}
 	}
 }
 \caption{\textbf{Client Side}}
 \label{alg:client}
\end{algorithm}	
\end{minipage}
\begin{minipage}{0.5\textwidth}
\begin{algorithm}[H]
\scriptsize 
\SetAlgoLined
 \textbf{Initialization:}\\
 $s\leftarrow$ socket initialized \;
 \textit{s.connect}\\
 $\eta, n, N, E \leftarrow s.synchronize()$\\
 $ \{\boldsymbol{w}^{( i)}, \boldsymbol{b}^{( i)}\}_{\forall i\in \{0..l\}} \ \leftarrow \ initialize\ using\ \Phi $\\
 $\displaystyle \ \ \ \  \{\mathbf{z}^{( i)}\}_{\forall i\in \{l+1..L\}} \leftarrow \emptyset \ $\\
 $\displaystyle \ \ \ \  \left\{\frac{\partial J}{\partial \mathbf{z}^{( i)}}\right\}_{\forall i\in \{l+1..L\}} \leftarrow \emptyset \ $\\
 \For{$\displaystyle e \ \in \ E $}{
 	\For{$\displaystyle i \leftarrow 1 \ \mathbf{to} \ N \ $}{
 	$\displaystyle \mathbf{Forward\ propagation:}$\\
 	$\displaystyle \ \ \ \ O.zero\_grad()  $\\
 	$\displaystyle \ \ \ \ s.receive\ (\mathbf{a}^{(l)}) \ \ $ \\
$\displaystyle \ \ \ \ \mathbf{a}^{(L)} \ \leftarrow \ f^{( i)}\left( \mathbf{a}^{(l)}\right)$\\
$\displaystyle \ \ \ \ s.send\left( \mathbf{a}^{(L)}\right)$\\
$\displaystyle \mathbf{Backward\ propagation:}$\\
$\displaystyle \ \ \ \ s.receive \ \left( \frac{\partial J}{\partial \mathbf{a}^{(L)}}\right)$\\
$\displaystyle \ \ \ \ \text{Compute}\ \left\{ \frac{\partial J}{\partial \boldsymbol{w}^{(L)}}, \ \frac{\partial J}{\partial \boldsymbol{b}^{(L)}} \right\}$\\
$\displaystyle \ \ \ \ \text{Update}\ \boldsymbol{w}^{( L)},\ \boldsymbol{b}^{(L)}  $\\
$\displaystyle \ \ \ \ \text{Compute}\ \frac{\partial J}{\partial \mathbf{a}^{( l)}} $\\
$\displaystyle \ \ \ \ s.send \left( \frac{\partial J}{\partial \mathbf{a}^{( l)}} \right)$\\
 		}
 }
 \caption{\textbf{Server Side}}
 \label{alg:server}
\end{algorithm}
\end{minipage}

Further information on the HE parameters and how to choose the best-suited parameters can be found in the~\href{https://bit.ly/3KY8ByN}{TenSEAL's benchmarks tutorial}. As shown in~\autoref{alg:clientHE} and~\autoref{alg:serverHE}, the context is either $\mathsf{ctx_{pub}}$ or $\mathsf{ctx_{pri}}$ depending on whether it holds the secret key $\mathsf{sk}$. Both the $\mathsf{ctx_{pub}}$ and $\mathsf{ctx_{pri}}$ have the same parameters, though $\mathsf{ctx_{pri}}$ holds a $\mathsf{sk}$ and $\mathsf{ctx_{pub}}$ does not. The server does not have access to the $\mathsf{sk}$ as the client only shares the $\mathsf{ctx_{pub}}$ with the server. 

\paragraph*{Forward propagation} In the forward propagation 
the client first zeroes out the gradients for the batch of data $(\mathbf{x}, \mathbf{y})$. He then begins calculating the $\mathbf{a}^{(l)}$ AMs from $\mathbf{x}$, as can be seen in~\autoref{alg:clientHE} where each $f^{(i)}$ is a Conv1D layer.

The~\href{https://pytorch.org/docs/stable/generated/torch.nn.Conv1d.html}{Conv1D} layer can be described as following: given a 1D input signal that contains $C$ channels, where each channel $\mathbf{x}_{(i)}$ is a 1D array ($i\in \{1,\ldots,C\}$), a Conv1D layer produces an output that contains $C'$ channels. The $j^{th}$ output channel $\mathbf{y}_{(j)}$, where $j\in \{1,\ldots,C'\}$, can be described as
\begin{equation}\label{eq:1dconvOp}
\scriptsize
	\mathbf{y}_{(j)} = \boldsymbol{b}_{(j)}  + \sum_{i=1}^{C} \boldsymbol{w}_{(i)} \star \mathbf{x}_{(i)},
\end{equation}
where $\boldsymbol{w}_{(i)}, i\in \{1,\ldots,C\}$ are the weights, $\boldsymbol{b}_{(j)}$ are the biases of the Conv1D layer, and $\star$ is the 1D cross-correlation operation. The $\star$ operation can be described as
\begin{equation}
\scriptsize
	\mathbf{z}(i) = (\boldsymbol{w} \star \mathbf{x}) (i) = \sum_{j=0}^{m-1}\boldsymbol{w}(j) \cdot \mathbf{x}(i+j), 
\end{equation}
where $\mathbf{z}(i)$ denotes the $i^{th}$ element of $\mathbf{z}$, and 
size of the 1D weighted kernel is $m$. 

In~\autoref{alg:clientHE}, $g^{(i)}$ can be seen as the combination of Max Pooling and Leaky ReLU functions. The final output AMs of the $l^{th}$ layer from the client is $\mathbf{a}^{(l)}$. The client then homomorphically encrypts $\mathbf{a}^{(l)}$ and sends the EAMs $\overbar{\mathbf{a}^{(l)}}$ to the server. In~\autoref{alg:serverHE}, the server receives $\overbar{\mathbf{a}^{(l)}}$ and then performs forward propagation, which is a linear layer evaluated on HE encrypted data $\overbar{\mathbf{a}^{(l)}}$ as
\begin{equation}
	\label{eq:serverHELinear}
\scriptsize
 \overbar{\mathbf{a}^{(L)}} = \overbar{\mathbf{a}^{(l)}} \boldsymbol{w}^{(L)} + \boldsymbol{b}^{(L)} .
\end{equation}
Upon reception, the client decrypts $\overbar{\mathbf{a}^{(L)}}$ to get $\mathbf{a}^{(L)}$, performs Softmax on $\mathbf{a}^{(L)}$ to produce the predicted output $\mathbf{\hat{y}}$ and calculate the loss $J$, as can be seen in~\autoref{alg:clientHE}. 
Having finished the forward propagation 
we may move on to the backward propagation part of the protocol.

\paragraph*{Backward propagation} After calculating the loss $J$, the client starts the backward propagation by initially computing 
$\frac{\partial J}{\partial \hat{\mathbf{y}}}$ and then $\frac{\partial J}{\partial \mathbf{a}^{(L)}}$ and $ \frac{\partial J}{\partial \boldsymbol{w}^{(L)}}$ using the chain rule~(\autoref{alg:clientHE}). 
Specifically, the client calculates:
\begin{align}
\label{equ:chainrule}
\scriptsize
	\frac{\partial J}{\partial \mathbf{a}^{(L)}} &= \frac{\partial J}{\partial \hat{\mathbf{y}}} \frac{\partial \hat{\mathbf{y}}}{\partial \mathbf{a}^{(L)}},~~
	\frac{\partial J}{\partial \boldsymbol{w}^{(L)}} = \frac{\partial J}{\partial \mathbf{a}^{(L)}} \frac{\partial \mathbf{a}^{(L)}}{\partial \boldsymbol{w}^{(L)}}
 \end{align}

Following, the client sends $\frac{\partial J}{\partial \mathbf{a}^{(L)}}$ and $ \frac{\partial J}{\partial \boldsymbol{w}^{(L)}}$ to the server. Upon reception, the server computes $\frac{\partial J}{\partial \boldsymbol{b}}$ by simply doing $\frac{\partial J}{\partial \boldsymbol{b}}= \frac{\partial J}{\partial \mathbf{a}^{(L)}}$, based on equation~\eqref{eq:serverHELinear}. The server then updates weights and biases of linear layer according to equation~\eqref{equ:serverUpdateWB}.

\begin{align}
	\label{equ:serverUpdateWB}
 \scriptsize
	\boldsymbol{w}^{(L)} =  \boldsymbol{w}^{(L)} - \eta\frac{\partial J}{\partial \boldsymbol{w}^{(L)}}, \quad & b^{(L)} = \boldsymbol{b}^{(L)} - \eta\frac{\partial J}{\partial \boldsymbol{b}^{(L)}}
\end{align}

\noindent Next, the server calculates
\begin{equation}
\scriptsize
	\frac{\partial J}{\partial \mathbf{a}^{(l)}} = \frac{\partial J}{\partial \mathbf{a}^{(L)}} \frac{\partial \mathbf{a}^{(L)}}{\partial \mathbf{a}^{(l)}},
\end{equation}
and sends $\frac{\partial J}{\partial \mathbf{a}^{(l)}}$ to the client. After receiving $\frac{\partial J}{\partial \mathbf{a}^{(l)}}$, the client calculates the gradients of $J$ w.r.t the weights and biases of the Conv1D layer using the chain-rule, which can generally be described 
as:
\begin{align}
	\label{equ: gradients}
 \scriptsize
	\frac{\partial J}{\partial \boldsymbol{w}^{(i-1)}} &= \frac{\partial J}{\partial \boldsymbol{w}^{(i)}}\frac{\partial \boldsymbol{w}^{(i)}}{\partial \boldsymbol{w}^{(i-1)}}, ~~
	\frac{\partial J}{\partial \boldsymbol{b}^{(i-1)}} = \frac{\partial J}{\partial \boldsymbol{b}^{(i)}}\frac{\partial \boldsymbol{b}^{(i)}}{\partial \boldsymbol{b}^{(i-1)}}    
\end{align}

Finally, after calculating the gradients $\frac{\partial J}{\partial \boldsymbol{w}^{(i)}}, \ \frac{\partial J}{\partial \boldsymbol{b}^{(i)}}$, the client updates $\boldsymbol{w}^{(i)}$ and $\boldsymbol{b}^{(i)}$ using the Adam optimization algorithm~\cite{kingma2015adam}.

\section{Formal Protocol Construction}	
\label{sec:protocol}

In this section, we formalize the communication between the client and the server. 
To this end, we design a protocol that is divided in two phases (Setup and Running) and relies on the following five building blocks:

\begin{minipage}{0.5\textwidth}
\begin{algorithm}[H]
\scriptsize 
\SetAlgoLined
 \textbf{Context Initialization:}\\
 $\displaystyle \ \ \ \ \ \mathsf{ctx_{pri}},\ \leftarrow \ \mathcal{P}, \ \mathcal{C}, \ \Delta, \  \mathsf{pk}, \ \mathsf{sk}$\\
 $\displaystyle \ \ \ \ \ \mathsf{ctx_{pub}},\ \leftarrow \ \mathcal{P}, \ \mathcal{C}, \ \Delta, \  \mathsf{pk}$\\
			 $\displaystyle \ \ \ \ \ s.send( \mathsf{ctx_{pub}})$\\
	\For{$\displaystyle e \ \text{in} \ E $}{
	\For{$\displaystyle \text{ each} \ \text{batch}\ ( \mathbf{x},\ \mathbf{y}) \ \text{from}\ \mathbf{D}\ $}{
	$\displaystyle  \mathbf{Forward\ propagation:}$\\
	$\displaystyle \ \ \ \ O.zero\_grad()  $\\
	$\displaystyle \ \ \ \ \mathbf{a}^{0}\ \ \leftarrow \mathbf{x}$ \\
	\For{$i\ \leftarrow \ 1\ \mathbf{to} \ l$}{
	$\displaystyle \ \ \ \ \ \ \ \ \mathbf{z}^{( i)} \ \leftarrow \ f^{( i)}\left( \mathbf{a}^{( i-1)}\right)$\\
	$\displaystyle \ \ \ \ \ \ \ \ \mathbf{a}^{i} \ \leftarrow \ g^{( i)}\left( \mathbf{z}^{( i)}\right)$}
	$\displaystyle \ \ \ \ \overbar{\mathbf{a}^{(l)}} \ \leftarrow \ \mathsf{HE.Enc}\left(\mathsf{pk}, \mathbf{a}^{(l)}\right)$\\
	$\displaystyle \ \ \ \ s.send \ \overbar{(\mathbf{a}^{(l)})}$\\
	$\displaystyle \ \ \ \  s.receive\ ( \overbar{\mathbf{a}^{(L)})}$\\
	$\displaystyle  \ \ \ \  \mathbf{a}^{( L)} \ \leftarrow \ \mathsf{HE.Dec}\left(\mathsf{sk}, \overbar{\mathbf{a}^{( L)}}\right)$\\
	$\displaystyle  \ \ \ \  \hat{\mathbf{y}} \ \leftarrow \ Softmax\left(\mathbf{a}^{( L)}\right)$\\
	$\displaystyle \ \ \ \  \mathbf{J} \leftarrow \mathcal{L} (\hat{\mathbf{y}}, \mathbf{y})$\\
	$\displaystyle \mathbf{Backward\ propagation:}$\\
	$\displaystyle \ \ \ \ \text{Compute}\left\{\frac{\partial J}{\partial \hat{\mathbf{y}}} \& \frac{\partial J}{\partial \mathbf{a}^{(L)} } \& \frac{\partial J}{\partial \boldsymbol{w}^{(L)}} \right\}$\\
	$\displaystyle \ \ \ \ s.send\left(\frac{\partial J}{\partial \mathbf{a}^{(L)} } \& \frac{\partial J}{\partial \boldsymbol{w}^{(L)}}\right)$\\
	$\displaystyle \ \ \ \ s.receive\left( \frac{\partial J}{\partial \mathbf{a}^{(l)}} \right)$\\
	\For{$i\leftarrow l\ \text{down to} \ 1$}{
	$\displaystyle \ \ \ \ \ \ \ \ \text{Compute} \left\{ \frac{\partial J}{\partial \boldsymbol{w}^{(i)}}, \ \frac{\partial J}{\partial \boldsymbol{b}^{(i)}} \right\}$\\
	$\displaystyle \ \ \ \ \ \ \ \ \text{Update}\ \boldsymbol{w}^{( i)},\ \boldsymbol{b}^{(i)} $}
	}	
	}
 \caption{\textbf{Client Side}}
 \label{alg:clientHE}
\end{algorithm}
\end{minipage}
\begin{minipage}{0.5\textwidth}
\begin{algorithm}[H]
\scriptsize 
\SetAlgoLined
 \textbf{Context Initialization:}\\
 $\displaystyle \ \ \ \ \ s.receive(\mathsf{ctx_{pub}})$\\
 \For{$\displaystyle \mathbf{ e} \ \text{in} \ \mathbf{E} $}{
 \For{$\displaystyle i \leftarrow 1 \ \mathbf{to} \ N \ $ }{
 $\displaystyle \mathbf{Forward\ propagation:}$\\
 $\displaystyle \ \ \ \ O.zero\_grad()  $\\
 $\displaystyle \ \ \ \ s.receive\ \overbar{(\mathbf{a}^{(l)})}$\\
 $\displaystyle \ \ \ \ \overbar{\mathbf{a}^{(L)}} \ \leftarrow \ \mathsf{HE.Eval} \left(f^{( i)}\left( \overbar{\mathbf{a}^{( l)}}\right)\right)$\\
 $\displaystyle \ \ \ \ s.send\left( \overbar{\mathbf{a}^{( L)}}\right)$\\
 $\displaystyle \mathbf{Backward\ propagation:}$\\
 $\displaystyle \ \ \ \ s.receive\left\{\frac{\partial J}{\partial \mathbf{a}^{(L)}} \& \frac{\partial J}{\partial \boldsymbol{w}^{(L)}}\right\}$\\
 $\displaystyle \ \ \ \ \text{Compute}\ \frac{\partial J}{\partial \boldsymbol{b}^{(L)}} $\\
 $\displaystyle \ \ \ \  \text{Update}\ \boldsymbol{w}^{(L)},\ \boldsymbol{b}^{(L)} $\\
 $\displaystyle \ \ \ \ \text{Compute}\ \frac{\partial J}{\partial \mathbf{a}^{(l)}} $\\
 $\displaystyle \ \ \ \ s.send \left( \frac{\partial J}{\partial \mathbf{a}^{(l)}} \right)$\\
 }
 }
 \caption{\textbf{Server Side}}
 \label{alg:serverHE}
\end{algorithm}
\end{minipage}

\begin{itemize}
	\item A CCA2 secure public-key encryption scheme ${\mathsf{PKE} = \left(\mathsf{Gen, Enc, Dec}\right)}$;
	\item An EUF-CMA secure signature scheme ${\mathsf{Sign} = (\mathsf{\sigma, ver})}$;
	\item A Leveled Homomorphic Encryption scheme ${\mathsf{HE} = (\mathsf{KeyGen, Enc, Eval, Dec})}$;
	\item A first and second pre-image resistant hash function $H$;
	\item A synchronized clock between the Client and the Server.
\end{itemize}


\paragraph*{Setup Phase} During this phase each entity generates a public/private key pair $\mathsf{(pk, sk)}$ for the CCA2-secure public-key encryption scheme $\mathsf{PKE}$ and a sign/verification key pair $\mathsf{(\sigma, ver)}$ for the EUF-CMA-secure signature scheme $\mathsf{Sign}$. Furthermore, the client runs $\mathsf{HE.KeyGen}$ to generate the public, private and evaluation key of the HE scheme. Below we provide a list of all generated keys:
\begin{itemize}
	\item $\mathsf{(pk_C, sk_C)}$ - public/private key pair for the Client;
	\item $\mathsf{(\sigma_{C}, ver_{C})}$ - sign/verification key pair for the Client;
	\item $\mathsf{(pk_S, sk_S)}$ - public/private key pair for the Server;
	\item $\mathsf{(\sigma_{S}, ver_{S})}$ - sign/verification key pair for the Server;
	\item $\mathsf{(pk_{HE}, sk_{HE}, evk_{HE})}$ - public, private and evaluation keys generated by Client.
\end{itemize}

\paragraph*{Running Phase} After 
successfully executing the \textbf{Setup} phase, the Client initiates the protocol's running phase 
by sending 

$m_1 = \langle t_1, \mathsf{PKE.Enc}(\mathsf{pk_{S}}, \mathsf{evk_{HE}}), \mathsf{HE.Enc}(\mathsf{pk_{HE}}, \mathrm{AMap}), \sigma_{C}(H_1) \rangle$ 
to the server, where $t_1$ is a timestamp, $H_1$ is a hash such that: 

$H_1 = H(\mathsf{evk_{HE}}\|\mathsf{HE.Enc}(\mathsf{pk_{HE}}, \mathrm{AMap}))$, $\mathrm{AMap}$ is the AM and $\sigma_C(\cdot)$ denotes the cryptographic signature of $C$. Upon reception, the Server verifies the freshness of the message by looking at the timestmap $t_1$ and the signature of the sender. If any of the verifications fail, the Server outputs $\perp$ and aborts the protocol. Otherwise, it first decrypts the evaluation key $\mathsf{evk}_{HE}$ using its private key $\mathsf{sk_S}$ and subsequently uses the homomorphic evaluation key $\mathsf{evk_{HE}}$ to operate on $\mathrm{Amap}$. The result of these operations is a homomorphically encrypted output $\mathsf{HE.Enc(pk_{HE}, out)}$, 
which is sent back to the Client via: $m_2 = \langle t_2, \mathsf{HE.Enc(pk_{HE}, out)}, \sigma_{S}(H_2) \rangle$, where $H(2) = H(t_2 \| \mathsf{HE.Enc(pk_{HE}, out)})$. Upon reception, the Client first verifies the freshness of the message and the server's signature. 
Should the verification fail, the Client outputs $\perp$ and aborts the protocol. Otherwise, the Client first recovers the encrypted output by running $\mathsf{HE.Dec}(\mathsf{sk}_{\mathsf{HE}}, \mathsf{HE.Enc(pk_{\mathsf{HE}}, out)}) \rightarrow \mathsf{out}$. This output will be used by the Server to compute the initial gradients $\mathsf{grad}$ for the ML model. Having computed $\mathsf{grad}$ the Client forwards them to the Server via: $m_3 = \langle t_3, \mathsf{PKE.Enc}(\mathsf{pk}_{C}, \mathsf{grad}), \sigma_{S}(H_3)\rangle$, where $H_3 = H(t_3\|\mathsf{grad})$. Upon receiving $m_3$, the Server first verifies the freshness and the signature of the message. Should the verification fail, the Server outputs $\perp$ and aborts the protocol. Otherwise, it recovers the gradients by running $\mathsf{PKE.Dec(sk_{S}, PKE.Enc(pk_{S}, grad)})\rightarrow \mathsf{grad}$. Based on $\mathsf{grad}$, the Server can update the parameters of the ML model (i.e. bias and weights), a process 
resulting to 
updated gradients $\mathsf{grad'}$. Finally, the Server outsources $\mathsf{grad}'$ to the Client via: $m_4 = \langle t_4, \mathsf{PKE.Enc(pk_{C}, grad'}), \sigma_{C}(H_4)\rangle$, where $H_4 = H(t_4\|\mathsf{grad}'))$. Upon reception, Server verifies freshness and message signature. Should the verification fail, Client outputs $\perp$ and aborts the protocol. Otherwise it decrypts the updated gradients by running $\mathsf{PKE.Dec(sk_{C}, \mathsf{PKE.Enc(pk_{C}, grad'}))} \rightarrow \mathsf{grad'}$. The running phase of our protocol is illustrated in \autoref{fig:protocol}.
%
    
   

\section{Protocol Security}
\label{app:secanalysis}
We prove the security of our protocol in presence of a probabilistic polynomial time (PPT) adversary $\mathcal{ADV}$. We assume that $\mathcal{ADV}$ has the following capabilities:

\begin{itemize}
	\item $\mathcal{ADV}$ overhears the communication between the Client and the Server;
	\item $\mathcal{ADV}$ is allowed to tamper with any message she sees, either by changing the contents of the message, or by replacing it with another one.
\end{itemize} 

In our threat model, we assume that $\mathcal{ADV}$ does not block the communication between the Client and the Server.

More formally, we will prove the following proposition:

\begin{proposition}[Protocol Soundness]
Let $\mathsf{PKE}$ be a CCA2-secure public-key encryption scheme and $\mathsf{Sign}$ an EUF-CMA-secure signature scheme. Moreover, let $\mathcal{ADV}$ be a PPT adversary. Then $\mathcal{ADV}$:
\begin{enumerate}
	\item Can \textbf{not} infer any information from the exchanged messages except from the time of sending;
	\item Can \textbf{not} tamper with the content of any message in a way that goes unnoticed.
\end{enumerate}
\end{proposition}

\begin{proof}
We examine each assumptions separately, and 
we will prove that they both hold with overwhelming probability.

\begin{enumerate}[label=\bfseries A\arabic*:]
	\item Our first assumption is that $\mathcal{ADV}$ can not infer any information from exchanged messages. Assuming that $\mathcal{ADV}$ does not collude with neither Client nor Server, the only way to infer information about messages is to successfully decrypt messages encrypted under $\mathsf{PKE}$ or $\mathsf{HE}$. However, assuming security of both of those schemes, this can only happen with negligible probability in the security parameters $\lambda$ and $\kappa$ of $\mathsf{PKE}$ and $\mathsf{HE}$ respectively. Hence, if we denote advantage of $\mathcal{ADV}$ in decrypting exchanged ciphertexts by $\epsilon$, we get: $ \epsilon = negl(\lambda) + negl(\kappa)$.
So finite sum of negligible functions is still negligible, $\mathcal{ADV}$ can decrypt messages with negligible probability and, our assumption holds with overwhelming probability.
	\item $\mathcal{ADV}$ can try tampering with the exchanged messages in two possible ways:
		\begin{itemize}
			\item Generate and send her own messages in place of the actual messages;
			\item Replay old messages.
		\end{itemize}
Generating her own valid ciphertexts is trivial as every ciphertext is encrypted under a public key. Moreover, $\mathcal{ADV}$ would also need to forge a valid signature of the sender should she wish to create a malicious message that is indistinguishable from a real one. However, given the EUF-CMA security of the signature scheme $\mathsf{Sign}$, this can only happen with negligible probability in the security parameter $\mu$ of $\mathsf{Sign}$. 

Thus, the only alternative for $\mathcal{ADV}$ is to replay and old message that was transmitted at some time in the past $t'$. However, since, each exchange message contains the current timestamp both in the first component of the message and in the signed hash, $\mathcal{ADV}$ would once again need to forge a valid signature of the sender since, otherwise, the verification of the signature would pass, though the verification of the timestamp would fail. Hence, if we denote by $\epsilon_2$ the advantage of $\mathcal{ADV}$ in forging a valid signature, we conclude that the overall advantage of $\mathcal{ADV}$ in tampering with the content of any message in an indistinguishable way is:

\begin{equation}
\scriptsize
	\epsilon_2 = negl(\mu) + negl(\mu) = negl(\mu)
\end{equation}

Hence, our second assumption holds with overwhelming probability.

\end{enumerate}
\end{proof}

\section{Performance Analysis}
\label{sec:performance}

In this work, we evaluate our method on two ECG datasets: the MIT-BIH dataset~\cite{moody2001impact} and the PTB-XL dataset~\cite{wagner2020ptb}.
\footnote{Due to limited space, the figures from this section are moved to~\autoref{sec: datasets}.}


\subsection{Experimental Setup}
\label{subsec:setup}
All models are trained on a machine with Ubuntu 20.04 LTS, processor Intel Core i7-8700 CPU at 3.20GHz, 32Gb RAM, GPU GeForce GTX 1070 Ti with 8Gb of memory. We write our program in the~\href{https://www.python.org/downloads/release/python-397/}{Python version 3.9.7}. The NNs are constructed using the~\href{https://pytorch.org/get-started/previous-versions/}{PyTorch library version 1.8.1+cu102}. For HE algorithms, we employ the~\href{https://github.com/OpenMined/TenSEAL}{TenSeal library version 0.3.10}. 


In terms of hyperparameters, we train all networks with 10 epochs, a $\eta=0.001$ learning rate, and a $n=4$ training \textit{n}. 
For split NN with HE AMs, we use the Adam optimizer for client model and mini-batch Gradient Descent for server model. We use GPU for 
networks 
trained on plaintext. For U-shaped SL on HE AMs, we train the client model on GPU, and server model on CPU.

\subsection{Evaluation}
\label{subsec:experiments}
In this section, we report the experimental results in terms of accuracy, training duration and communication throughput. We measure the accuracy of the three NN on the plaintext test set after the training processes are completed.

\paragraph{\textbf{Networks with Different Activation Map Sizes:}}
The 1D CNN models used on both MIT-BIH and PTB-XL datasets have two Conv1D layers and one linear layer. The AMs are the output of the last Conv1D layer. 

For the MIT-BIH dataset, we experiment with two sizes of AMs: 
$[\text{\textit{n}}, 512]$ (as in~\cite{abuadbba2020can}) and $[\text{\textit{n}}, 256]$~\cite{khan2023split}. We get the AMs of size $[\textit{n}, 256]$ by reducing the number of output channels in the second Conv1D layer by half (from 16 output channels to 8 output channels). We denote the 1D CNN model with an AM sized $[\textit{n}, 256]$ as $M_1$, and the model with an AM sized $[\textit{n}, 512]$ as $M_2$.

For the PTB-XL dataset, we change the number of the input channels for the first Conv1D layer to 12, since the input data are 12-lead ECG signals~\cite{khan2023love}. Besides, we only experiment with 8 output channels for the second Conv1D layer. We denote this network by $M_3$. Using $M_3$, the output AM size is $[\textit{n}, 2000]$.

\paragraph{\textbf{Training Locally:}}
The result when training the model $M_1$ locally on the MIT-BIH plaintext dataset is shown in \autoref{fig:localTrain256}. The NN 
learns quickly and is able to decrease the loss drastically from epoch 1 to 5. After that, from epoch 6-10, the loss begins to plateau. After training for 10 epochs, we test the trained NN on the test dataset and get 88.06\% accuracy. Training the model locally on plaintext takes 4.8 seconds for each epoch on average.


\autoref{fig:localTrain512} shows the results when training the model $M_2$ on plaintext MIT-BIH. After 10 epochs, the model achieves the best training accuracy of 91.66\%. 
the trained model results in 92.84\% prediction accuracy. Each training epoch takes 4.8 seconds on average
-the same as model $M_1$. Even though the two models differ in the output sizes of the AMs, they are relatively small models: the number of parameters is 2061 for $M_1$, and 3989 for $M_2$. As both models can 
fit 
in the GPU memory, 
the local training duration becomes similar for both models.


Training $M_3$ locally on the plaintext PTB-XL dataset results in \autoref{fig:localTrainPTBXL} 
achieves a training accuracy of 72.65\% after 10 epochs with a test accuracy is 67.68\%. This low accuracy is due to small NN 
 with~12013 trainable parameters and limited training epochs (each takes~10.56 seconds on average).


\subsection{U-shaped Split Learning using Plaintext Activation Maps}
\label{subsec:plaintextactivationmap}
Our experiments, 
show that training the U-shaped split model on plaintext (reported in section~\ref{subsec:u_shaped_split}) 
produces the same results in terms of accuracy compared to local training 
for both models $M_1$, $M_2$ and $M_3$. This result is similar to the findings of 
~\cite{abuadbba2020can}. Even though the authors of~\cite{abuadbba2020can} only used the vanilla split model, they also found that compared to training locally, accuracy was not reduced.

We will now discuss the training time and communication overhead of the U-shaped split models and compare them to their local versions. For the split version of $M_1$, each training epoch takes 8.56 seconds on average, 
hence 43.9\% longer than local training. Training the split version of $M_2$ takes 8.67 seconds per epoch on average, which is 44.6\% longer compared to the 4.8 seconds of local training. 
The split version of $M_3$ on the PTB-XL dataset takes 15.55 seconds per epoch to train and is 47.25\% slower than the local version. 
The U-shaped split models take longer to train due to the communication between the client and the server. The communication cost for one epoch of training split $M_1$ and $M_2$ are 33.06 Mb and 60.12 Mb, respectively. $M_2$ incurs almost twice as much communication overhead compared to $M_1$ due to the bigger size of the AMs. For $M_3$, communication overhead is on average 316.9 Mb per epoch, which is much bigger than $M_1$ and $M_2$ due to bigger AMs sent from client during training.

These figures show 
that the similarities are not as strong compared to the MIT-BIH AMs. This is due to the fact that $M_3$ is a small network trained on a 
limited number of epochs (72.65\% training accuracy after 10 epochs), therefore, the convolution layers are not yet able to produce highly similar AMs compared to the original signals. Still the figures indicate 
similar patterns between the AMs and the original PTB-XL signals. Nonetheless, to reach any conclusions, we first need to experiment with better NNs that can produce highly accurate predictions on both the train and test splits of the PTB-XL dataset, then quantify similarities between AMs and input signals of these networks. 
\subsection{Visual Invertibility}
\label{subsec:VisualInvert}
In SL, certain AMs sent from client to the server to continue the training process show high similarity with the client's input data, as demonstrated in \autoref{fig:visual_invertibility} for models trained on the MIT-BIH dataset. The figure indicates that, compared to the raw input data from the client (the first row of~\autoref{fig:visual_invertibility}), some AMs (as plotted in the second and third row of~\autoref{fig:visual_invertibility}) have exceedingly similar patterns. This phenomenon clearly compromises the privacy of the client's raw data. The authors of~\cite{abuadbba2020can} quantify the privacy leakage by measuring the correlations between the AMs and the raw input signal by using two metrics: distance correlation and Dynamic Time Warping. This approach allows them to measure whether their solutions mitigate privacy leakage work. Since our work uses HE,
said metrics are unnecessary as the AMs are encrypted. Similar to MIT-BIH, we visualize the output AMs produced by $M_3$ to access their visual similarity compared to original signals. Due to space constraints, we visualize results only for normal class (see \autoref{fig:visual_invertibility_norm}) instead of five different classes of heartbeat in PTB-XL dataset.

\subsection{U-shaped Split 1D CNN with Encrypted Activation Maps}
\label{subsec:UShapedHE}
We train the split NNs $M_1$ and $M_2$ on the MIT-BIH dataset using EAMs according to~\ref{subsubsection: encryptedactivation}. To encrypt the AMs from the client before sending them to the server, we experiment with five different sets of HE parameters for both models $M_1$ and $M_2$.
Furthermore, we also employ the batch encryption (BE) feature of the CKKS encryption scheme. 
BE allows us to encrypt a $N \times N$ matrix into $N$ ciphertexts, with each column encrypted as a ciphertext for memory and computation optimization~\cite{benaissa2021tenseal}. We experiment with training the NNs with and without BE. Additionally, we perform experiments using different combinations of HE parameters. \autoref{tab:trainingTestingResultsMITBIH} shows the results in terms of training time, testing accuracy, and communication overhead for the NNs with different configurations on the MIT-BIH dataset. 
For the U-shaped SL version on the plaintext, we captured all communication between client and server. For training split models on EAMs, we approximate the communication overhead for one training epoch by getting the average communication of training on first ten batches of data, multiply that with total number of training batches.

Results differ between training $M_1$ and $M_2$ with different sets of HE parameters. For the $M_1$ model, the best test accuracy was 85.41\%, when using the set of HE parameters with $\mathcal{P}=4096$, $\mathcal{C}=[40, 20, 20]$, $\Delta=2^{21}$ (denoted $s_1$), and without BE. The accuracy drop was $2.65\%$ compared to plaintext training.

\begin{table*}[!h]
	\centering
	\caption{Training and testing results on the MIT-BIH dataset}
	\resizebox{\textwidth}{!}{%
		\begin{tabular}{c|c|cccc|c|c|c}
			\hline
			\multirow{2}{*}{Network} & \multirow{2}{*}{Type of Network} & \multicolumn{4}{c|}{HE Parameters}                                                                                      & \multirow{2}{*}{Training duration per epoch (s)} & \multirow{2}{*}{Test accuracy (\%)} & \multirow{2}{*}{Communication per epoch (Tb)} \\ \cline{3-6}
			&                                  & \multicolumn{1}{c|}{BE}                     & \multicolumn{1}{c|}{$\mathcal{P}$}    & \multicolumn{1}{c|}{$\mathcal{C}$}             & $\Delta$        &                                    &                                     &                                     \\ \hline
			\multirow{12}{*}{$M_1$}     & Local                            & \multicolumn{4}{c|}{}                                                                                                   & 4.80                                   & 88.06                               & 0                                   \\ \cline{2-9} 
			& Split (plaintext)                & \multicolumn{4}{c|}{}                                                                                                   & 8.56                                   & 88.06                               & 33.06e-6                            \\ \cline{2-9} 
			& \multirow{10}{*}{Split (HE)}     & \multicolumn{1}{c|}{\multirow{5}{*}{False}} & \multicolumn{1}{c|}{8192} & \multicolumn{1}{c|}{[60,40,40,60]} & $2^{40}$ & 50 318                                 & 85.31                               & 37.84                               \\ \cline{4-9} 
			&                                  & \multicolumn{1}{c|}{}                       & \multicolumn{1}{c|}{8192} & \multicolumn{1}{c|}{[40,21,21,40]} & $2^{21}$ & 48 946                                 & 80.63                               & 22.42                               \\ \cline{4-9} 
			&                                  & \multicolumn{1}{c|}{}                       & \multicolumn{1}{c|}{4096} & \multicolumn{1}{c|}{[40,20,20]}    & $2^{21}$ & 14 946                                 & 85.41                               & 4.49                                \\ \cline{4-9} 
			&                                  & \multicolumn{1}{c|}{}                       & \multicolumn{1}{c|}{4096} & \multicolumn{1}{c|}{[40,20,40]}    & $2^{20}$ & 18 129                                 & 80.78                               & 4.57                                \\ \cline{4-9} 
			&                                  & \multicolumn{1}{c|}{}                       & \multicolumn{1}{c|}{2048} & \multicolumn{1}{c|}{[18,18,18]}    & $2^{16}$ & 5 018                                  & 22.65                               & 0.58                                \\ \cline{3-9} 
			&                                  & \multicolumn{1}{l|}{\multirow{5}{*}{True}}  & \multicolumn{1}{c|}{8192} & \multicolumn{1}{c|}{[60,40,40,60]} & $2^{40}$ & 33 310                                 & 81.22                               & 4.77                                \\ \cline{4-9} 
			&                                  & \multicolumn{1}{l|}{}                       & \multicolumn{1}{c|}{8192} & \multicolumn{1}{c|}{[40,21,21,40]} & $2^{21}$ & 31 311                                 & 84.36                               & 2.81                                \\ \cline{4-9} 
			&                                  & \multicolumn{1}{l|}{}                       & \multicolumn{1}{c|}{4096} & \multicolumn{1}{c|}{[40,20,20]}    & $2^{21}$ & 11 507                                 & 79.00                               & 0.67                                \\ \cline{4-9} 
			&                                  & \multicolumn{1}{l|}{}                       & \multicolumn{1}{c|}{4096} & \multicolumn{1}{c|}{[40,20,40]}    & $2^{20}$ & 11 656                                 & 80.79                               & 0.69                                \\ \cline{4-9} 
			&                                  & \multicolumn{1}{l|}{}                       & \multicolumn{1}{c|}{2048} & \multicolumn{1}{c|}{[18,18,18]}    & $2^{16}$ & 3 869                                  & 70.12                               & 0.16                                \\ \hline
			\multirow{12}{*}{$M_2$}     & Local                            & \multicolumn{4}{c|}{}                                                                                                   & 4.80                                   & 92.84                               & 0                                   \\ \cline{2-9} 
			& Split (plaintext)                & \multicolumn{4}{c|}{}                                                                                                   & 8.67                                   & 92.84                               & 60.12e-6                            \\ \cline{2-9} 
			& \multirow{10}{*}{Split (HE)}     & \multicolumn{1}{c|}{\multirow{5}{*}{False}} & \multicolumn{1}{c|}{8192} & \multicolumn{1}{c|}{[60,40,40,60]} & $2^{40}$ & 118 518                                & 81.40                               & 238.71                              \\ \cline{4-9} 
			&                                  & \multicolumn{1}{c|}{}                       & \multicolumn{1}{c|}{8192} & \multicolumn{1}{c|}{[40,21,21,40]} & $2^{21}$ & n/a                                    & n/a                                 & n/a                                 \\ \cline{4-9} 
			&                                  & \multicolumn{1}{c|}{}                       & \multicolumn{1}{c|}{4096} & \multicolumn{1}{c|}{[40,20,20]}    & $2^{21}$ & 31 711                                 & 81.38                               & 12.86                               \\ \cline{4-9} 
			&                                  & \multicolumn{1}{c|}{}                       & \multicolumn{1}{c|}{4096} & \multicolumn{1}{c|}{[40,20,40]}    & $2^{20}$ & 31 791                                 & 80.12                               & 14.60                               \\ \cline{4-9} 
			&                                  & \multicolumn{1}{c|}{}                       & \multicolumn{1}{c|}{2048} & \multicolumn{1}{c|}{[18,18,18]}    & $2^{16}$ & 12 087                                 & 22.65                               & 1.786                               \\ \cline{3-9} 
			&                                  & \multicolumn{1}{c|}{\multirow{5}{*}{True}}  & \multicolumn{1}{c|}{8192} & \multicolumn{1}{c|}{[60,40,40,60]} & $2^{40}$ & 79 637                                 & 81.46                               & 17.57                               \\ \cline{4-9} 
			&                                  & \multicolumn{1}{c|}{}                       & \multicolumn{1}{c|}{8192} & \multicolumn{1}{c|}{[40,21,21,40]} & $2^{21}$ & 56 356                                 & 84.46                               & 9.25                                \\ \cline{4-9} 
			&                                  & \multicolumn{1}{c|}{}                       & \multicolumn{1}{c|}{4096} & \multicolumn{1}{c|}{[40,20,20]}    & $2^{21}$ & 20 790                                 & 84.69                               & 1.82                                \\ \cline{4-9} 
			&                                  & \multicolumn{1}{c|}{}                       & \multicolumn{1}{c|}{4096} & \multicolumn{1}{c|}{[40,20,40]}    & $2^{20}$ & 20 521                                 & 81.65                               & 1.82                                \\ \cline{4-9} 
			&                                  & \multicolumn{1}{c|}{}                       & \multicolumn{1}{c|}{2048} & \multicolumn{1}{c|}{[18,18,18]}    & $2^{16}$ & 8 113                                  & 73.82                               & 0.36                                \\ \hline
		\end{tabular}%
	}
	\label{tab:trainingTestingResultsMITBIH}
\end{table*}
However, with BE, $s_1$ produced only $79\%$ accuracy. 
Compared to the bigger sets of parameters with $\mathcal{P}=8192$, $s_1$ achieves higher accuracy while requiring much lower training time and communication overhead. The result when using the first set of parameters with $\mathcal{P}=8192$ is close ($85.31\%$), but with a much longer training time (3.67 times longer) and communication overhead (8.43 times higher). We observe that in some cases, training with BE results in better testing accuracy, while in some other, it leads to accuracy reduction. On the other hand, training on EAMs with BE is 23-35\% faster. The amount of communication overhead is also significantly reduced (up to 13 times). Interestingly, using the HE parameters with $\mathcal{P}=2048$ and without BE drastically reduces accuracy to $22.65\%$, however, with BE, it only reduces accuracy to $70.12\%$. Hence, the effect of BE in NN training needs further study.

Although $M_2$ achieves better accuracy than $M_1$ on plaintext data, it does not provide better results in encrypted version. The best accuracy of $M_2$ on encrypted version is $84.69\%$ using $s_1$ and with BE. Compared to the configuration that achieves best results for $M_1$ at 85.41\%, best configuration for $M_2$ takes 1.39 times more to train but incurs less communication overhead (about 2.5 times less). This is because the best configuration for $M_2$ occurs when BE is used, and best configuration for $M_1$ occurs when BE is not used. In general, training $M_2$ with EAMs takes 2-3 times longer to train and 3-6 times more communication overhead compared to $M_1$ with the same HE configuration. 

\begin{table*}
\centering
\caption{Training and testing results on the PTB-XL dataset}
\label{tab:trainingTestingResultsPTBXL}
\resizebox{\textwidth}{!}{%
\begin{tabular}{c|c|l|c|c|c|c|c|c} 
\hline
\multirow{2}{*}{Network} & \multirow{2}{*}{Type of Network} & \multicolumn{4}{c|}{HE Parameters}                                   & \multirow{2}{*}{Training duration per epoch (s)} & \multirow{2}{*}{Test accuracy (\%)} & \multirow{2}{*}{Communication per epoch (Tb)}  \\ 
\cline{3-6}
                         &                                  & \multicolumn{1}{c|}{BE} & $\mathcal{P}$ & $\mathcal{C}$   & $\Delta$ &                                    &                                     &                                    \\ 
\hline
\multirow{6}{*}{$M_3$}   & Local                            & \multicolumn{4}{c|}{}                                                & 10.56                                  & 67.68                               & 0                                    \\ 
\cline{2-9}
                         & Split (plaintext)                & \multicolumn{4}{c|}{}                                                & 15.55                                  & 67.68                               & 316.9e-6                           \\ 
\cline{2-9}
                         & \multirow{4}{*}{Split (HE)}      & \multirow{3}{*}{True}   & 8192          & {[}40,21,21,40] & $2^{21}$   & 72 534                                 & 65.42                               & 115.64                              \\ 
\cline{4-9}
                         &                                  &                         & 4096          & {[}40,20,20]    & $2^{21}$   & 24 061                                 & 64.22                               & 18.20                               \\ 
\cline{4-9}
                         &                                  &                         & 4096          & {[}40,20,40]    & $2^{20}$   & 22 570                                 & 65.23                               & 18.77                               \\ 
\cline{4-9}
                         &                                  &                         & 2048          & {[}18,18,18]    & $2^{16}$   & 7 605                                  & 65.33                               & 1.93                                \\
\hline
\end{tabular}
}
\end{table*}

The results of training different settings of $M_3$ on the PTB-XL dataset are reported in \autoref{tab:trainingTestingResultsPTBXL}. We only train the split version of $M_3$ with an EAM using BE. The HE set of parameters $\mathcal{P}=8192$, $\mathcal{C}=[40, 21, 21, 40]$, $\Delta=2^{21}$ achieves the best test accuracy at 65.42\%. This result is only 2.26\% lower than the result obtained by the plaintext version. However, this set of parameters incurs the most communication overhead (115.64 Tb per epoch) and takes the longest to train (72 534 seconds per epoch). Overall, test accuracies achieved by different HE parameters are quite close to each other, with 65.42\% being the lowest. Interestingly, the smallest set of HE parameters $\mathcal{P}=2048, \mathcal{C}=[18, 18, 18], \Delta=2^{16}$ achieves second-best accuracy at 65.33\%, requiring 
about 1/10 of the training duration and 1/100 of the communication overhead compared to $\mathcal{P}=8192$. The two sets of parameters with $\mathcal{P}=4096$ produce quite similar results, 
roughly taking same amount of time and communication overhead to train.

Through our experiments, we see that training on EAMs can produce very optimistic results, with accuracy dropping by 2-3\% for the best sets of HE parameters. Furthermore, training using BE can significantly reduce the amount of training time and communication overhead needed, while producing comparable results 
when it come to training without BE. The set of parameters with $\mathcal{P}=8192$ always achieve the highest test accuracy, though incurring the highest communication overhead and the longest training time. The set of parameters with $\mathcal{P}=4096$ can offer a good trade-off as they can produce on-par accuracy with $\mathcal{P}=8192$, while requiring 
significantly less communication and training time. Experimental results show that with the smallest set of HE parameters $\mathcal{P}=2048$, $\mathcal{C}=[18, 18, 18]$, $\Delta=2^{16}$, 
the least amount of communication and training time is required. In addition, this only works well when used together with BE. When training the network $M_3$ on the PTB-XL dataset, this set of parameters produces even better test accuracy compared to $\mathcal{P}=4096$. However, this result may be because the network $M_3$ is 
small. The test accuracy on the plaintext version is $67.68\%$, hence the noises produced by the HE algorithm 
do not yet have a significant role in reducing the model's accuracy. 

\begin{remark}
As can be seen in \autoref{tab:trainingTestingResultsMITBIH}, the accuracy of the same algorithm varies greatly under different CKKS parameters. 
The parameter selection in CKKS is not 
evident as 
a set of the parameter may result in efficient computation for one application 
but also in poor performance for another application. In addition, CKKS uses approximate arithmetic 
rather than exact arithmetic, in the sense that once 
computation is finished 
the result may slightly differ compared to that of a direct computation 
~\cite{clet2021bfv}. Hence, it is still open to research whether 
for 
specific applications, a closed form of relation for the set of parameters can be used to measure the accuracy. 
However, training both models $M_{1}$ and $M_{2}$ with a different set of parameters, we observe:
\begin{itemize}
    \item Training model with BE yields better outcomes than without BE. This pattern can be seen in $M_{2}$  and also in $M_{1}$ with two exceptions ($\mathcal{P}=8192, \mathcal{C}=[60, 40, 40, 60], \Delta=2^{40}$) and ($\mathcal{P}=4096, \mathcal{C}=[40, 20, 20], \Delta=2^{21}$).
    \item Training without BE, higher contexts yield 
    better results than lower. 
    This pattern can be seen in $M_2$ and 
    also in $M_{1}$ with two exceptions $(\mathcal{P}=4096, \mathcal{C}=[40, 20, 20], \Delta=2^{21})$ and $(\mathcal{P}=4096, \mathcal{C}=[40, 20, 40], \Delta=2^{20})$.
    \item Also in both models $M_{1}$ and $M_{2}$, BE is 
    suitable for lower context $(\mathcal{P}=2048, \mathcal{C}=[18, 18, 18], \Delta=2^{16})$.
\end{itemize}
\end{remark}

\noindent \textbf{Open Science and Reproducible Research:} 
To support open science and reproducible research, and provide researchers with opportunity to use, test, and extend our work, source code used for the evaluations is publicly available\footnote{\href{https://github.com/khoaguin/HESplitNet}{https://github.com/khoaguin/HESplitNet}}.
\section{Conclusion}
\label{sec:Conclusion}
In this paper, we focused on training ML models in a privacy-preserving way. We used the concept of SL in combination with HE and constructed protocols allowing a client to train a model in collaboration with a server without sharing valuable information about the raw data. To the best of our knowledge, this is the first work that uses SL on encrypted data. Our experiments show that our approach has achieved high accuracy, especially when compared with less secure approaches that combine SL with differential privacy. The limitation of our work is having only one client in the protocol. While extending the protocol to multiple clients is an important task, it requires us to rely on a multi-key HE scheme, which is beyond scope of this work and remains to be addressed in future works.

\bibliographystyle{splncs04}
\bibliography{main}

\begin{thebibliography}{10}
\providecommand{\url}[1]{\texttt{#1}}
\providecommand{\urlprefix}{URL }
\providecommand{\doi}[1]{https://doi.org/#1}

\bibitem{abuadbba2020can}
Abuadbba, S., Kim, K., Kim, M., Thapa, C., Camtepe, S.A., Gao, Y., Kim, H.,
  Nepal, S.: Can we use split learning on 1d cnn models for privacy preserving
  training? In: Proceedings of the 15th ACM Asia Conference on Computer and
  Communications Security. pp. 305--318 (2020)

\bibitem{benaissa2021tenseal}
Benaissa, A., Retiat, B., Cebere, B., Belfedhal, A.E.: Tenseal: A library for
  encrypted tensor operations using homomorphic encryption. In: Workshop on
  Distributed and Private Machine Learning. ICLR (2021)

\bibitem{cabrero2021sok}
Cabrero-Holgueras, J., Pastrana, S.: Sok: Privacy-preserving computation
  techniques for deep learning. Proceedings on Privacy Enhancing Technologies
  \textbf{2021}(4),  139--162 (2021)

\bibitem{cheon2017homomorphic}
Cheon, J.H., Kim, A., Kim, M., Song, Y.: Homomorphic encryption for arithmetic
  of approximate numbers. In: International Conference on the Theory and
  Application of Cryptology and Information Security. pp. 409--437. Springer
  (2017)

\bibitem{clet2021bfv}
Clet, P.E., Stan, O., Zuber, M.: Bfv, ckks, tfhe: Which one is the best for a
  secure neural network evaluation in the cloud? In: International Conference
  on Applied Cryptography and Network Security. pp. 279--300. Springer (2021)

\bibitem{gentry2009fully}
Gentry, C.: Fully homomorphic encryption using ideal lattices. In: Proceedings
  of the forty-first annual ACM symposium on Theory of computing. pp. 169--178
  (2009)

\bibitem{gupta2018distributed}
Gupta, O., Raskar, R.: Distributed learning of deep neural network over
  multiple agents. Journal of Network and Computer Applications  \textbf{116},
  ~1--8 (2018)

\bibitem{hitaj2017deep}
Hitaj, B., Ateniese, G., Perez-Cruz, F.: Deep models under the gan: information
  leakage from collaborative deep learning. In: Proceedings of the 2017 ACM
  SIGSAC Conference on Computer and Communications Security. pp. 603--618
  (2017)

\bibitem{khan2021blind}
Khan, T., Bakas, A., Michalas, A.: Blind faith: Privacy-preserving machine
  learning using function approximation. In: 2021 IEEE Symposium on Computers
  and Communications (ISCC). pp.~1--7. IEEE (2021)

\bibitem{khan2023split}
Khan, T., Nguyen, K., Michalas, A.: Split ways: Privacy-preserving training of
  encrypted data using split learning. In: Fletcher, G., Kantere, V. (eds.)
  Proceedings of the Workshops of the {EDBT/ICDT} 2023 Joint Conference,
  Ioannina, Greece, March, 28, 2023. {CEUR} Workshop Proceedings, vol.~3379.
  CEUR-WS.org (2023),
  \url{https://ceur-ws.org/Vol-3379/HeDAI\_2023\_paper402.pdf}

\bibitem{khan2023love}
Khan, T., Nguyen, K., Michalas, A., Bakas, A.: Love or hate? share or split?
  privacy-preserving training using split learning and homomorphic encryption.
  In: The 20th Annual International Conference on Privacy, Security \& Trust
  (PST2023) 21-23 August, 2023, Copenhagen, Denmark (2023)

\bibitem{kingma2015adam}
Kingma, D.P., Ba, J.: Adam: A method for stochastic optimization. In:
  International Conference on Learning Representations (ICLR). ICLR (2015)

\bibitem{li2017classification}
Li, D., Zhang, J., Zhang, Q., Wei, X.: Classification of ecg signals based on
  1d convolution neural network. In: 2017 IEEE 19th International Conference on
  e-Health Networking, Applications and Services (Healthcom). pp.~1--6. IEEE
  (2017)

\bibitem{maas2013rectifier}
Maas, A.L., Hannun, A.Y., Ng, A.Y., et~al.: Rectifier nonlinearities improve
  neural network acoustic models. In: Proceedings of the 30th International
  Conference on Machine Learning,. vol.~28, p.~3. Citeseer (2013)

\bibitem{mcmahan2017communication}
McMahan, B., Moore, E., Ramage, D., Hampson, S., y~Arcas, B.A.:
  Communication-efficient learning of deep networks from decentralized data.
  In: Artificial intelligence and statistics. pp. 1273--1282. PMLR (2017)

\bibitem{moody2001impact}
Moody, G.B., Mark, R.G.: The impact of the mit-bih arrhythmia database. IEEE
  Engineering in Medicine and Biology Magazine  \textbf{20}(3),  45--50 (2001)

\bibitem{paillier1999public}
Paillier, P.: Public-key cryptosystems based on composite degree residuosity
  classes. In: International conference on the theory and applications of
  cryptographic techniques. pp. 223--238. Springer (1999)

\bibitem{maxpooling2010}
Scherer, D., Müller, A., Behnke, S.: Evaluation of pooling operations in
  convolutional architectures for object recognition. Artificial Neural
  Networks - ICANN 2010 - 20th International Conference pp. 92--101 (01 2010).
  \doi{10.1007/978-3-642-15825-4_10}

\bibitem{vepakomma2019reducing}
Vepakomma, P., Gupta, O., Dubey, A., Raskar, R.: Reducing leakage in
  distributed deep learning for sensitive health data. In: AI for Social Good
  Workshop. ICLR (2019)

\bibitem{vepakomma2019split}
Vepakomma, P., Gupta, O., Swedish, T., Raskar, R.: Split learning for health:
  Distributed deep learning without sharing raw patient data. In: AI for social
  good workshop. ICLR (2019)

\bibitem{wagner2020ptb}
Wagner, P., Strodthoff, N., Bousseljot, R.D., Kreiseler, D., Lunze, F.I.,
  Samek, W., Schaeffter, T.: Ptb-xl, a large publicly available
  electrocardiography dataset. Scientific data  \textbf{7}(1),  1--15 (2020)

\bibitem{yang2019federated}
Yang, Q., Liu, Y., Cheng, Y., Kang, Y., Chen, T., Yu, H.: Federated learning.
  Synthesis Lectures on Artificial Intelligence and Machine Learning
  \textbf{13}(3),  1--207 (2019)

\end{thebibliography}

\appendix

\section{Preliminaries}
\label{sec:preliminaries}
\subsection{Convolutional Neural Network}
In this work, we employ a 1D CNN~\cite{li2017classification,abuadbba2020can} as a feature detector and classifier for two ECG heartbeat datasets, namely MIT-BIH~\cite{moody2001impact} and PTB-XL~\cite{wagner2020ptb}. The employed 1D CNN has the following stacked layers:
\begin{itemize}
    \item \underline{Conv1D}: 
    Is used to swipe a kernel of adjustable weights over a 1D input signal. The 
    Conv1D outputs the AMs 
    capturing 
    feature information from said input signal. \autoref{fig:1dcnn} visualizes the Conv1D operation and 
    shows the difference between a Conv1D and a Conv2D layer.
    \item \underline{Leaky ReLU~\cite{maas2013rectifier}}: 
    It is a non-linear function that can be described as $f(x) = x$, if ${x\geq 0}$, and $f(x)=\alpha x$ if $x<0$, where $\alpha$ is a small number, such as 0.01.
    \item \underline{Max Pooling}: In CNN, Max Pooling compresses the input, focusing on important elements and allowing slight input changes with minimal impact on the pooled version~\cite{maxpooling2010}.
    
    \item \underline{Fully Connected (FC)}: The FC layer has one output unit connected to all 
    input units, unlike 
    the convolution layer, where one output unit in an AM only connects to a small area in the input signal. 
    \item \underline{Softmax}: A Softmax activation function takes a vector of $k$ real numbers, e.g. $\mathbf{z}=(z_1,\ldots,z_k)\mathbb{R}^k$, and outputs a probability distribution consisting of $k$ probabilities. Each probability in the output vector is calculated as follows: 
    \begin{equation}
        \sigma(\mathbf{z})_i = \frac{e^{z_i}}{\sum_{j=1}^{k} e^{z_j}}.
    \end{equation}
\end{itemize}
We use the 1D CNN as a supervised learning method, where both the input data and corresponding labels are needed to train our network.

\subsection{Homomorphic Encryption}
HE is an emerging cryptographic technique for computations on encrypted data. HE schemes are divided into three main categories according to their functionality: Partial HE~\cite{paillier1999public}, leveled (or somewhat) HE~\cite{cheon2017homomorphic}, and fully HE~\cite{gentry2009fully}. Each scheme has its own benefits and disadvantages. In this work, we use the CKKS Leveled HE scheme~\cite{cheon2017homomorphic}. CKKS allows users to do additions and a limited number of multiplications on vectors of complex values (and hence, real values too). Prior to the encryption, CKKS encodes a message $\mathbf{z} \in \mathbb{C}^{N/2}$ to a ring of polynomials over the integers $\mathbb{Z}[X]/\left(X^N+1\right)$. Working with polynomials over rings of integers is a good trade-off between security and efficiency compared to standard computations on vectors. During encoding, the vector $\mathbf{z}$ is multiplied by a \textit{scaling factor} $\Delta$ to keep a level of precision. The encoded message is encrypted, and the resulted ciphertext is an element $c \in \left(\mathbb{Z}_q[X]/\left(X^N +1\right)\right)^2$. Ciphertexts can then be added or multiplied together. An issue arises during multiplication, is that the term $\Delta^2$ appears in the ciphertext result. To address this, CKKS deploys a rescaling operation to keep the scaling factor $\Delta$ constant. The inverse procedure needs to be followed for the decryption; that is, the ciphertext will be decrypted first, and then the encoded message will be decoded and multiplied by $1/\Delta$ to recover $\mathbf{z}' \in \mathbb{C}^{N/2}$. The number of allowed multiplications is predefined by a list of prime numbers. To build this list, the authors first choose $\left(p_1, \dots, p_L, q_0\right)$ primes, where each $p_{\ell} \approx \Delta$ and $q_0 > \Delta$. Finally, they set $q_L = \prod_{1}^L p_l\cdot q_0$, where $q_{L} = q$ -- the order of $\mathbb{Z}_q$ -- and the list is $(q_L, \dots, q_0)$. After each multiplication, an element is deleted from the list. However, according to~\cite{cheon2017homomorphic}, the security of CCKS is based on the ratio $N/q$. Hence, to maintain the same level of security as we increase $q$, we also need to increase $N$-the degree of the polynomials and hence, computational costs. 
Summing up, the most important parameters of the CKKS scheme are:
	
\begin{enumerate}[\bfseries 1.]
    \item \textbf{Polynomial Modulus $\mathcal{P}$}: Naturally, this parameter has a direct impact on the scheme's efficiency and security. According to~\cite{cheon2017homomorphic}, this value needs to be a power of two. Common values include 2048, 4096, 8192, 16384 and 32768.
    \item \textbf{Coefficient Modulus $\mathcal{C}$}: A list of primes that define current scheme's level. After each multiplication a different prime is used as coefficient modulus. Hence, no more multiplications are allowed when all primes are used. 
    \item \textbf{Scaling Factor $\Delta$}: This is a constant positive number multiplied by the plaintext message during encoding to maintain a certain level of precision.
\end{enumerate}	

 \begin{figure}
    \centering
    \begin{minipage}{0.45\textwidth}
        \centering
 \begin{adjustbox}{width=\textwidth}
		\tikzset{every picture/.style={line width=0.75pt}} 
\begin{tikzpicture}[x=0.75pt,y=0.75pt,yscale=-1,xscale=1]
\draw  [fill={rgb, 255:red, 255; green, 244; blue, 199 }  ,fill opacity=1 ] (167,155.9) -- (179.9,143) -- (210,143) -- (210,265.1) -- (197.1,278) -- (167,278) -- cycle ; \draw   (210,143) -- (197.1,155.9) -- (167,155.9) ; \draw   (197.1,155.9) -- (197.1,278) ;
\draw  [fill={rgb, 255:red, 200; green, 218; blue, 164 }  ,fill opacity=1 ] (197,165.6) -- (209.6,153) -- (239,153) -- (239,255.5) -- (226.4,268.1) -- (197,268.1) -- cycle ; \draw   (239,153) -- (226.4,165.6) -- (197,165.6) ; \draw   (226.4,165.6) -- (226.4,268.1) ;
\draw  [fill={rgb, 255:red, 218; green, 246; blue, 242 }  ,fill opacity=1 ] (225.2,166.8) -- (239,153) -- (271.2,153) -- (271.2,255.2) -- (257.4,269) -- (225.2,269) -- cycle ; \draw   (271.2,153) -- (257.4,166.8) -- (225.2,166.8) ; \draw   (257.4,166.8) -- (257.4,269) ;
\draw  [fill={rgb, 255:red, 197; green, 181; blue, 175 }  ,fill opacity=1 ] (104,153.9) -- (116.9,141) -- (147,141) -- (147,284.1) -- (134.1,297) -- (104,297) -- cycle ; \draw   (147,141) -- (134.1,153.9) -- (104,153.9) ; \draw   (134.1,153.9) -- (134.1,297) ;
\draw  [fill={rgb, 255:red, 255; green, 244; blue, 199 }  ,fill opacity=1 ] (258.4,153.9) -- (271.3,141) -- (301.4,141) -- (301.4,262.1) -- (288.5,275) -- (258.4,275) -- cycle ; \draw   (301.4,141) -- (288.5,153.9) -- (258.4,153.9) ; \draw   (288.5,153.9) -- (288.5,275) ;
\draw  [fill={rgb, 255:red, 200; green, 218; blue, 164 }  ,fill opacity=1 ] (289,165.6) -- (301.6,153) -- (331,153) -- (331,255.4) -- (318.4,268) -- (289,268) -- cycle ; \draw   (331,153) -- (318.4,165.6) -- (289,165.6) ; \draw   (318.4,165.6) -- (318.4,268) ;
\draw  [fill={rgb, 255:red, 218; green, 246; blue, 242 }  ,fill opacity=1 ] (318.4,163.98) -- (329.38,153) -- (355,153) -- (355,257.02) -- (344.02,268) -- (318.4,268) -- cycle ; \draw   (355,153) -- (344.02,163.98) -- (318.4,163.98) ; \draw   (344.02,163.98) -- (344.02,268) ;
\draw  [fill={rgb, 255:red, 195; green, 220; blue, 252 }  ,fill opacity=1 ] (429.4,149.9) -- (442.3,137) -- (472.4,137) -- (472.4,258.1) -- (459.5,271) -- (429.4,271) -- cycle ; \draw   (472.4,137) -- (459.5,149.9) -- (429.4,149.9) ; \draw   (459.5,149.9) -- (459.5,271) ;
\draw  [fill={rgb, 255:red, 197; green, 181; blue, 175 }  ,fill opacity=1 ] (284.5,291.9) -- (297.4,279) -- (327.5,279) -- (327.5,365.16) -- (314.6,378.06) -- (284.5,378.06) -- cycle ; \draw   (327.5,279) -- (314.6,291.9) -- (284.5,291.9) ; \draw   (314.6,291.9) -- (314.6,378.06) ;
\draw  [fill={rgb, 255:red, 246; green, 200; blue, 185 }  ,fill opacity=1 ] (315.6,291.9) -- (328.5,279) -- (358.6,279) -- (358.6,364.17) -- (345.7,377.07) -- (315.6,377.07) -- cycle ; \draw   (358.6,279) -- (345.7,291.9) -- (315.6,291.9) ; \draw   (345.7,291.9) -- (345.7,377.07) ;
\draw [line width=1.5]    (147,207) -- (163,207) ;
\draw [shift={(167,207)}, rotate = 180] [fill={rgb, 255:red, 0; green, 0; blue, 0 }  ][line width=0.08]  [draw opacity=0] (11.61,-5.58) -- (0,0) -- (11.61,5.58) -- cycle    ;
\draw [line width=1.5]    (356,207) -- (424,207) ;
\draw [shift={(428,207)}, rotate = 180] [fill={rgb, 255:red, 0; green, 0; blue, 0 }  ][line width=0.08]  [draw opacity=0] (11.61,-5.58) -- (0,0) -- (11.61,5.58) -- cycle    ;
\draw [line width=1.5]    (437,298) -- (362,298) ;
\draw [shift={(358,298)}, rotate = 360] [fill={rgb, 255:red, 0; green, 0; blue, 0 }  ][line width=0.08]  [draw opacity=0] (11.61,-5.58) -- (0,0) -- (11.61,5.58) -- cycle    ;
\draw  [dash pattern={on 0.84pt off 2.51pt}] (97,129) -- (376.5,129) -- (376.5,389) -- (97,389) -- cycle ;
\draw  [dash pattern={on 0.84pt off 2.51pt}] (385,129) -- (485.5,129) -- (485.5,388) -- (385,388) -- cycle ;
\draw [line width=1.5]    (437,271) -- (437,298) ;
\draw [line width=1.5]    (359,325) -- (453,325) ;
\draw [line width=1.5]    (452,325) -- (452,278) ;
\draw [shift={(452,274)}, rotate = 90] [fill={rgb, 255:red, 0; green, 0; blue, 0 }  ][line width=0.08]  [draw opacity=0] (11.61,-5.58) -- (0,0) -- (11.61,5.58) -- cycle    ;
\draw (174.42,260.07) node [anchor=north west][inner sep=0.75pt]  [rotate=-269.64] [align=left] {1D Convolutional};
\draw (265.42,260.07) node [anchor=north west][inner sep=0.75pt]  [rotate=-269.64] [align=left] {1D Convolutional};
\draw (202.42,253.07) node [anchor=north west][inner sep=0.75pt]  [rotate=-269.64] [align=left] {Leaky Relu};
\draw (295.42,249.07) node [anchor=north west][inner sep=0.75pt]  [rotate=-269.64] [align=left] {Leaky Relu};
\draw (235.42,249.07) node [anchor=north west][inner sep=0.75pt]  [rotate=-269.64] [align=left] {Max Pooling};
\draw (325.42,249.07) node [anchor=north west][inner sep=0.75pt]  [rotate=-269.64] [align=left] {Max Pooling};
\draw (161,361) node [anchor=north west][inner sep=0.75pt]   [align=left] {\textbf{Client-side}};
\draw (389,362) node [anchor=north west][inner sep=0.75pt]   [align=left] {\textbf{Server-side}};
\draw (325.42,350.83) node [anchor=north west][inner sep=0.75pt]  [rotate=-269.64] [align=left] {Softmax};
\draw (290.42,350.28) node [anchor=north west][inner sep=0.75pt]  [rotate=-269.64] [align=left] {Output};
\draw (115.42,280.07) node [anchor=north west][inner sep=0.75pt]  [rotate=-269.64] [align=left] {Training Input Data};
\draw (435.42,246.07) node [anchor=north west][inner sep=0.75pt]  [rotate=-269.64] [align=left] {Fully Connected};
\end{tikzpicture}
  \end{adjustbox}
		\caption{U-shaped SL}
		\label{fig:u-shapedSL}

\end{minipage}\hfill
    \begin{minipage}{0.5\textwidth}
        \centering
    \resizebox{\textwidth}{!}{%
        \tikzset{every picture/.style={line width=0.75pt}}       
        
        \tikzset{every picture/.style={line width=0.75pt}} 
        
        \begin{tikzpicture}[x=0.75pt,y=0.75pt,yscale=-1,xscale=1]
            
            \draw  [color={rgb, 255:red, 0; green, 0; blue, 0 }  ,draw opacity=1 ] (45.5,149) -- (148.5,149) -- (148.5,229) -- (45.5,229) -- cycle ;
            \draw  [color={rgb, 255:red, 74; green, 144; blue, 226 }  ,draw opacity=1 ] (80.5,22) -- (160,22) -- (160,97.5) -- (80.5,97.5) -- cycle ;
            \draw  [color={rgb, 255:red, 74; green, 144; blue, 226 }  ,draw opacity=1 ] (282.5,10) -- (359,10) -- (359,82.5) -- (282.5,82.5) -- cycle ;
            \draw  [color={rgb, 255:red, 0; green, 0; blue, 0 }  ,draw opacity=1 ] (463,99) -- (539.5,99) -- (539.5,146) -- (463,146) -- cycle ;
            \draw  [draw opacity=0] (244,95) -- (425.5,95) -- (425.5,255.5) -- (244,255.5) -- cycle ; \draw  [color={rgb, 255:red, 0; green, 0; blue, 0 }  ,draw opacity=1 ] (244,95) -- (244,255.5)(264,95) -- (264,255.5)(284,95) -- (284,255.5)(304,95) -- (304,255.5)(324,95) -- (324,255.5)(344,95) -- (344,255.5)(364,95) -- (364,255.5)(384,95) -- (384,255.5)(404,95) -- (404,255.5)(424,95) -- (424,255.5) ; \draw  [color={rgb, 255:red, 0; green, 0; blue, 0 }  ,draw opacity=1 ] (244,95) -- (425.5,95)(244,115) -- (425.5,115)(244,135) -- (425.5,135)(244,155) -- (425.5,155)(244,175) -- (425.5,175)(244,195) -- (425.5,195)(244,215) -- (425.5,215)(244,235) -- (425.5,235)(244,255) -- (425.5,255) ; \draw  [color={rgb, 255:red, 0; green, 0; blue, 0 }  ,draw opacity=1 ]  ;
            \draw  [draw opacity=0] (193,97) -- (213.5,97) -- (213.5,258.5) -- (193,258.5) -- cycle ; \draw  [color={rgb, 255:red, 0; green, 0; blue, 0 }  ,draw opacity=1 ] (193,97) -- (193,258.5)(213,97) -- (213,258.5) ; \draw  [color={rgb, 255:red, 0; green, 0; blue, 0 }  ,draw opacity=1 ] (193,97) -- (213.5,97)(193,117) -- (213.5,117)(193,137) -- (213.5,137)(193,157) -- (213.5,157)(193,177) -- (213.5,177)(193,197) -- (213.5,197)(193,217) -- (213.5,217)(193,237) -- (213.5,237)(193,257) -- (213.5,257) ; \draw  [color={rgb, 255:red, 0; green, 0; blue, 0 }  ,draw opacity=1 ]  ;
            \draw  [color={rgb, 255:red, 74; green, 144; blue, 226 }  ,draw opacity=1 ][dash pattern={on 4.5pt off 4.5pt}] (240,90) -- (310,90) -- (310,161.5) -- (240,161.5) -- cycle ;
            \draw  [color={rgb, 255:red, 74; green, 144; blue, 226 }  ,draw opacity=1 ][dash pattern={on 4.5pt off 4.5pt}] (359,187) -- (429,187) -- (429,258.5) -- (359,258.5) -- cycle ;
            \draw  [color={rgb, 255:red, 128; green, 128; blue, 128 }  ,draw opacity=1 ][fill={rgb, 255:red, 155; green, 155; blue, 155 }  ,fill opacity=1 ] (405,95.5) -- (424.5,95.5) -- (424.5,113.5) -- (405,113.5) -- cycle ;
            \draw [color={rgb, 255:red, 0; green, 0; blue, 0 }  ,draw opacity=1 ]   (310,127) .. controls (405.02,128.49) and (235.21,204.73) .. (356.65,208.45) ;
            \draw [shift={(358.5,208.5)}, rotate = 181.37] [color={rgb, 255:red, 0; green, 0; blue, 0 }  ,draw opacity=1 ][line width=0.75]    (10.93,-3.29) .. controls (6.95,-1.4) and (3.31,-0.3) .. (0,0) .. controls (3.31,0.3) and (6.95,1.4) .. (10.93,3.29)   ;
            \draw  [color={rgb, 255:red, 74; green, 144; blue, 226 }  ,draw opacity=1 ][dash pattern={on 4.5pt off 4.5pt}] (189,89) -- (217.5,89) -- (217.5,145.5) -- (189,145.5) -- cycle ;
            \draw  [color={rgb, 255:red, 74; green, 144; blue, 226 }  ,draw opacity=1 ][dash pattern={on 4.5pt off 4.5pt}] (189,209.5) -- (217.5,209.5) -- (217.5,261) -- (189,261) -- cycle ;
            \draw  [color={rgb, 255:red, 155; green, 155; blue, 155 }  ,draw opacity=1 ][fill={rgb, 255:red, 128; green, 128; blue, 128 }  ,fill opacity=1 ] (194,157.5) -- (213.5,157.5) -- (213.5,175.5) -- (194,175.5) -- cycle ;
            \draw [color={rgb, 255:red, 0; green, 0; blue, 0 }  ,draw opacity=1 ]   (202.5,145.5) -- (202.5,204.5) ;
            \draw [shift={(202.5,206.5)}, rotate = 270] [color={rgb, 255:red, 0; green, 0; blue, 0 }  ,draw opacity=1 ][line width=0.75]    (10.93,-3.29) .. controls (6.95,-1.4) and (3.31,-0.3) .. (0,0) .. controls (3.31,0.3) and (6.95,1.4) .. (10.93,3.29)   ;
            \draw [color={rgb, 255:red, 155; green, 155; blue, 155 }  ,draw opacity=1 ]   (423.5,102.5) .. controls (435.32,99.55) and (421.91,120.85) .. (460.69,120.53) ;
            \draw [shift={(462.5,120.5)}, rotate = 178.6] [color={rgb, 255:red, 155; green, 155; blue, 155 }  ,draw opacity=1 ][line width=0.75]    (10.93,-3.29) .. controls (6.95,-1.4) and (3.31,-0.3) .. (0,0) .. controls (3.31,0.3) and (6.95,1.4) .. (10.93,3.29)   ;
            \draw [color={rgb, 255:red, 155; green, 155; blue, 155 }  ,draw opacity=1 ]   (192.5,163.5) .. controls (167.54,171.18) and (161.93,180.7) .. (152.68,192.95) ;
            \draw [shift={(151.5,194.5)}, rotate = 307.57] [color={rgb, 255:red, 155; green, 155; blue, 155 }  ,draw opacity=1 ][line width=0.75]    (10.93,-3.29) .. controls (6.95,-1.4) and (3.31,-0.3) .. (0,0) .. controls (3.31,0.3) and (6.95,1.4) .. (10.93,3.29)   ;
            \draw [color={rgb, 255:red, 74; green, 144; blue, 226 }  ,draw opacity=1 ]   (185.5,104.5) .. controls (162.23,118.86) and (163.47,135.17) .. (113.04,99.59) ;
            \draw [shift={(111.5,98.5)}, rotate = 35.43] [color={rgb, 255:red, 74; green, 144; blue, 226 }  ,draw opacity=1 ][line width=0.75]    (10.93,-3.29) .. controls (6.95,-1.4) and (3.31,-0.3) .. (0,0) .. controls (3.31,0.3) and (6.95,1.4) .. (10.93,3.29)   ;
            \draw [color={rgb, 255:red, 74; green, 144; blue, 226 }  ,draw opacity=1 ]   (258.5,87.5) .. controls (298.1,57.8) and (242.63,79.07) .. (280.33,50.38) ;
            \draw [shift={(281.5,49.5)}, rotate = 143.13] [color={rgb, 255:red, 74; green, 144; blue, 226 }  ,draw opacity=1 ][line width=0.75]    (10.93,-3.29) .. controls (6.95,-1.4) and (3.31,-0.3) .. (0,0) .. controls (3.31,0.3) and (6.95,1.4) .. (10.93,3.29)   ;
            
            \draw (154,267) node [anchor=north west][inner sep=0.75pt]   [align=left] {\textcolor[rgb]{0,0,0}{1D Conv layer}};
            \draw (278,269) node [anchor=north west][inner sep=0.75pt]   [align=left] {\textcolor[rgb]{0,0,0}{2D Conv layer}};
            \draw (467,105) node [anchor=north west][inner sep=0.75pt]   [align=left] {\textcolor[rgb]{0,0,0}{A pixel in }\\\textcolor[rgb]{0,0,0}{an image}};
            \draw (89,32) node [anchor=north west][inner sep=0.75pt]  [color={rgb, 255:red, 74; green, 144; blue, 226 }  ,opacity=1 ] [align=left] {\textcolor[rgb]{0.29,0.56,0.89}{1D Kernel }\\\textcolor[rgb]{0.29,0.56,0.89}{Feature}\\\textcolor[rgb]{0.29,0.56,0.89}{Detector}};
            \draw (58,161) node [anchor=north west][inner sep=0.75pt]   [align=left] {\textcolor[rgb]{0,0,0}{A value in }\\\textcolor[rgb]{0,0,0}{the ECG }\\\textcolor[rgb]{0,0,0}{signal vector}};
            \draw (287,16) node [anchor=north west][inner sep=0.75pt]   [align=left] {\textcolor[rgb]{0.29,0.56,0.89}{2D Kernel }\\\textcolor[rgb]{0.29,0.56,0.89}{Feature}\\\textcolor[rgb]{0.29,0.56,0.89}{Detector}};
            
        \end{tikzpicture}
    }%
    \caption{Conv1D vs Conv2D}
    \label{fig:1dcnn}
\end{minipage} \hfill
    \begin{minipage}{0.45\textwidth}
        \centering
    \resizebox{\textwidth}{!}{%
  \begin{sequencediagram}
    \newthread{A}{Client}{}
    \newinst[9]{B}{Server}{}
    \begin{messcall}{A}{$m_1 = \langle t_1, \mathsf{PKE.Enc}(\mathsf{pk_{S}}, \mathsf{evk_{HE}}), \mathsf{HE.Enc}(\mathsf{pk_{HE}}, \mathrm{AMap}), \sigma_{C}(H_1) \rangle$}{B}{}
    
    \begin{callself}{B}{$\mathsf{HE.Eval}(\mathsf{pk}_{\mathsf{HE}}, \mathrm{AMap})$}{$\mathsf{HE.Enc(pk_{HE}, out)}$}
    \end{callself}
    \begin{messcall}{B}{$m_2 = \langle t_2, \mathsf{HE.Enc(pk_{HE}, out)}, \sigma_{S}(H_2) \rangle$}{A}{}
   
   \begin{callself}{A}{Use $\mathsf{out}$ to compute gradients}{$\mathsf{grad}$}
   \end{callself}
   \begin{messcall}{A}{$m_3 = \langle t_3, \mathsf{PKE.Enc}(\mathsf{pk}_{C}, \mathsf{grad}), \sigma_{S}(H_3)\rangle$}{B}{}
   \begin{callself}{B}{Update parameters}{$\mathsf{grad'}$}
   \end{callself}
   \begin{messcall}{B}{$m_4 = \langle t_4, \mathsf{PKE.Enc(pk_{C}, grad'}), \sigma_{C}(H_4)\rangle$}{A}{}
   \end{messcall}
   \end{messcall}
   \end{messcall} 
   \end{messcall}
  \end{sequencediagram}}
  \caption{Running Phase}
  \label{fig:protocol}
\end{minipage}

	\end{figure}
\section{Datasets}
\label{sec: datasets}
\paragraph{\textbf{MIT-BIH:}}
We use the pre-processed dataset from~\cite{abuadbba2020can}, which is based on the MIT-BIH arrhythmia 
database~\cite{moody2001impact}. The processed dataset contains 26,490 samples of heartbeat that belong to 5 different types: N (normal beat), left (L) and right (R) bundle branch block, atrial (A) and ventricular (V) premature contraction. 
To train our network, the dataset is split into a train and test split 
as matrices of size $[13245, 1, 128]$, meaning that each contain 13,245 ECG samples and, each sample has one channel and 128 timesteps~\cite{abuadbba2020can}.

\paragraph{\textbf{PTB-XL:}}
According to~\cite{wagner2020ptb}, PTB-XL is the largest open-source ECG dataset since 
2020. The dataset contains 12-lead ECG-waveforms from 21837 records of 18885 patients. Compared to PTB-XL, MIT-BIH only contains 2-lead ECG-waveforms obtained from 47 patients. Each waveform from PTB-XL 
has a duration of 10 seconds. 
Two sampling rates 
are used to collect the data: 100 Hz and 500 Hz. In our experiment, we employ the 100 Hz waveforms. Each 12-lead ECG waveform is associated with one or several classes out of five classes: normal (NORM), conduction disturbance (CD), myocardial infarction (MI), hypertrophy (HYP), and ST/T change (STTC). For waveforms that belong to multiple classes, we choose only the first one and remove the others for simplicity. 
The dataset is then split into a 90\%-10\% train-test ratio. In total, we have a training split of size $[19267, 12, 1000]$, 
with 19,267 ECG waveform samples, 
of 12 channels (or leads) and 1,000 timesteps each. The test split's size is 
$[2163, 12, 1000]$.

 \begin{figure}
    \centering
    \begin{minipage}{0.5\textwidth}
        \centering
        \includegraphics[width=0.9\textwidth]{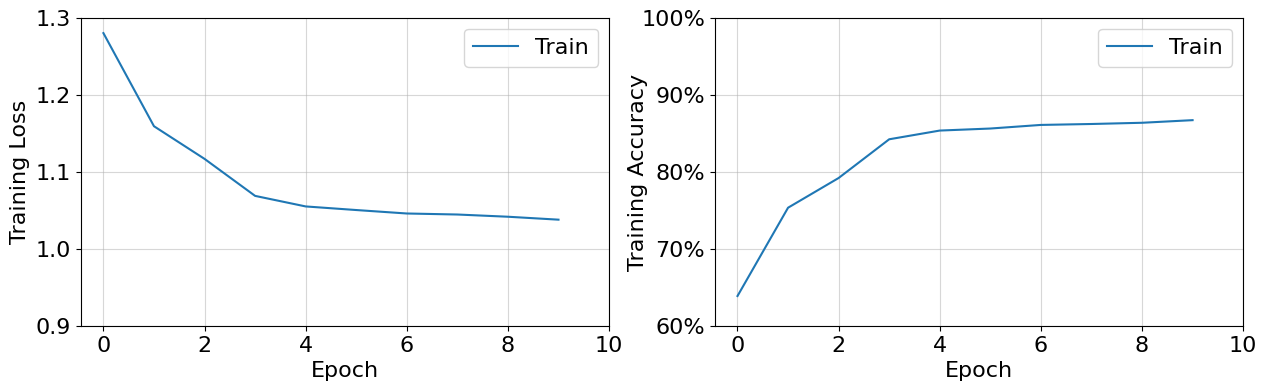} 
        \caption{ MIT-BIH $[\textit{n}, 256]$} \label{fig:localTrain256}
    \end{minipage}\hfill
    \begin{minipage}{0.49\textwidth}
        \centering
        \includegraphics[width=0.9\textwidth]{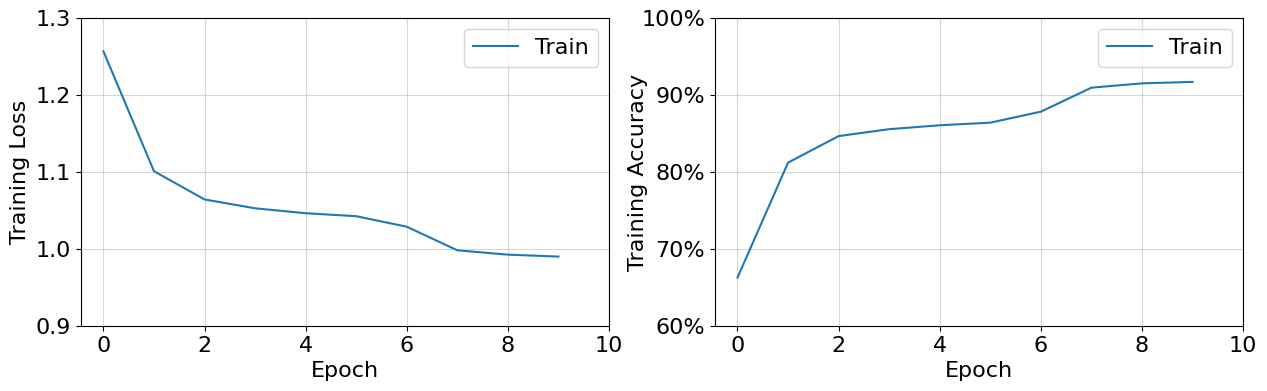} 
        \caption{MIT-BIH $[\textit{n}, 512]$} \label{fig:localTrain512}
    \end{minipage}\hfill
        \begin{minipage}{0.5\textwidth}
        \centering
        \includegraphics[width=0.9\textwidth]{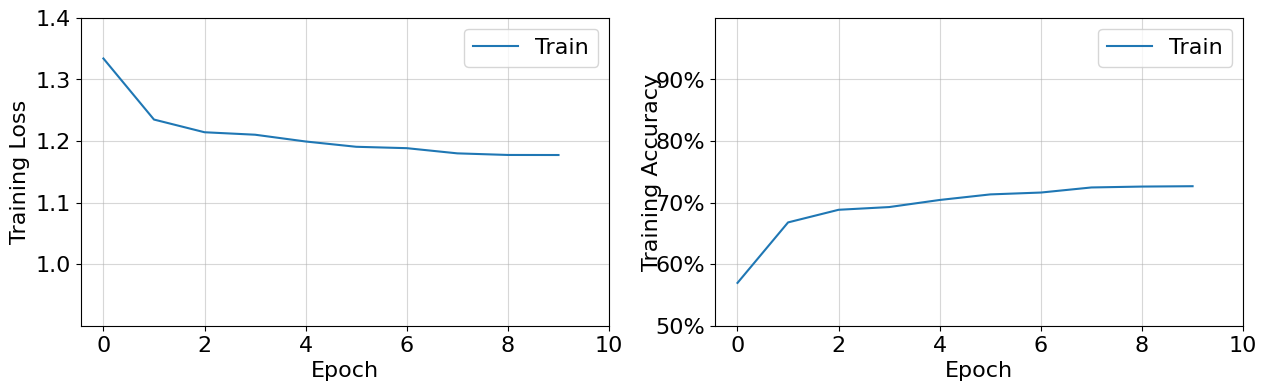} 
        \caption{PTB-XL $[\textit{n}, 512]$}
	\label{fig:localTrainPTBXL}
    \end{minipage}
\end{figure}

\begin{figure}
    \centering
    \begin{minipage}{0.45\textwidth}
        \centering
        \includegraphics[width=0.75\textwidth]{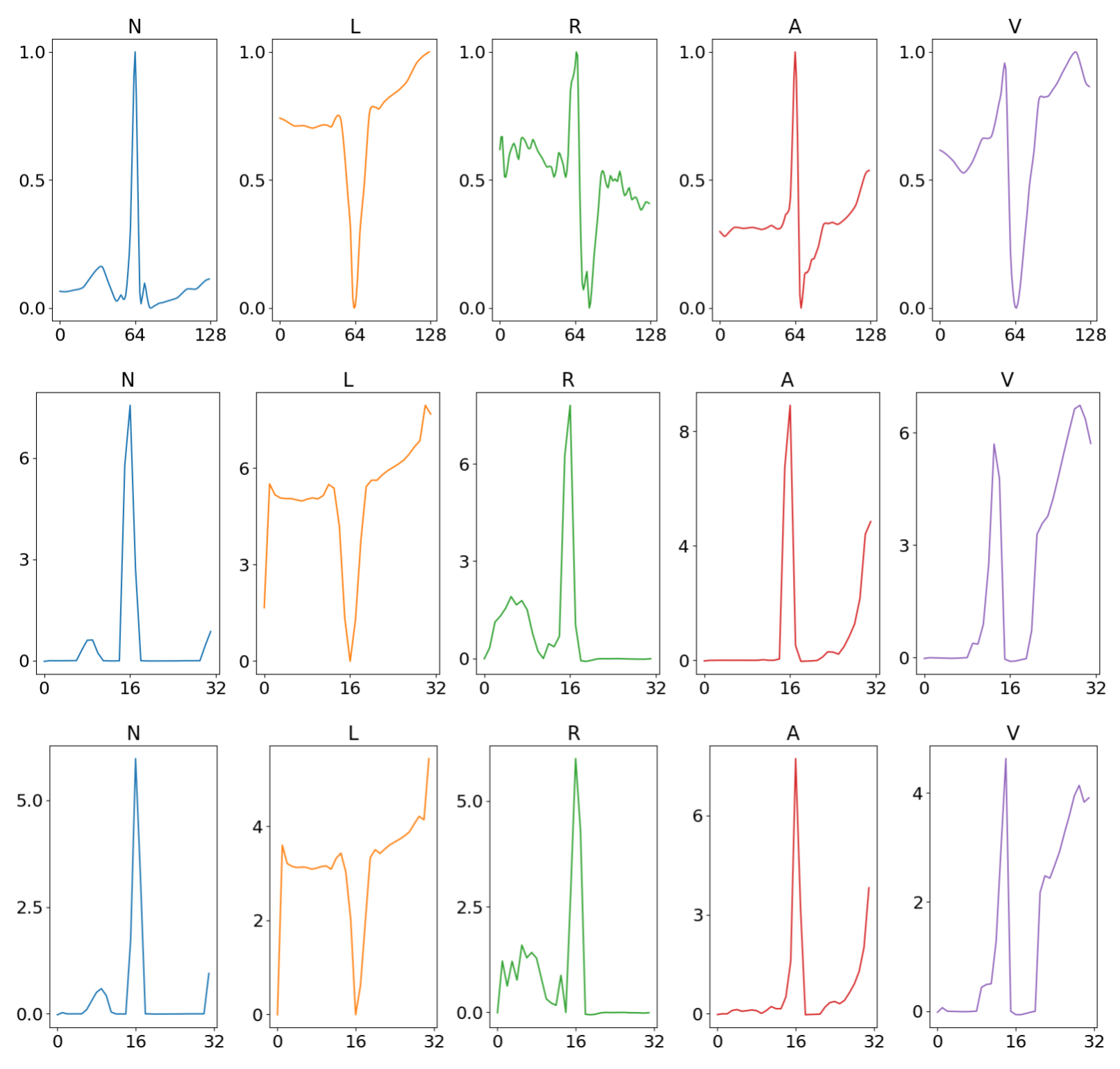} 
        \caption{Top: client input data. Middle: output channels $M_1$. 
	Bottom: output channels $M_2$.}
	\label{fig:visual_invertibility}
    \end{minipage}\hfill
    \begin{minipage}{0.54\textwidth}
        \centering
        \includegraphics[width=\textwidth]{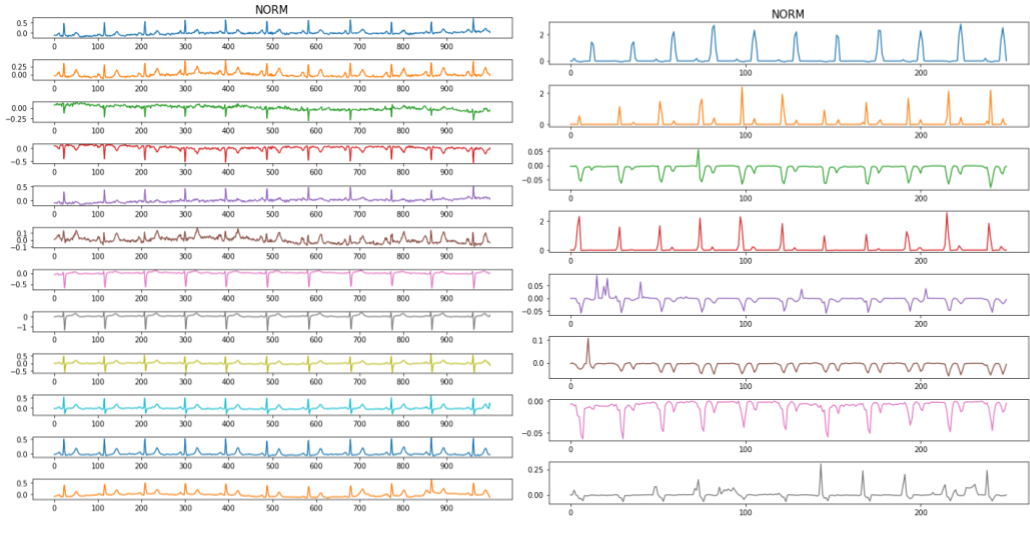}
	\caption{Visual invertibility of the model $M_3$ on the PTB-XL dataset. Left: input data (NORM class). Right: corresponding activation maps.}
	\label{fig:visual_invertibility_norm}
    \end{minipage}
\end{figure}

\end{document}